\documentclass[reqno]{amsart}
\usepackage{float}
\usepackage{amsmath}
\usepackage{graphicx}
\usepackage{latexsym}
\usepackage{amsfonts}
\usepackage{amssymb}
\usepackage{graphicx}
\usepackage{dsfont}
\usepackage[utf8]{inputenc}
\usepackage{wasysym}
\usepackage{xfrac}

\usepackage{enumitem}

\usepackage[shortcuts]{extdash}

\usepackage{romannum}

\usepackage[notref,notcite]{showkeys}

\usepackage{todo}

\theoremstyle{plain}
\newtheorem{theorem}{Theorem}
\newtheorem{corollary}[theorem]{Corollary}
\newtheorem{lemma}[theorem]{Lemma}
\newtheorem{proposition}[theorem]{Proposition}
\theoremstyle{definition}

\newtheorem{remark}[theorem]{Remark}
\newtheorem*{remark*}{Remark}

\DeclareMathOperator{\Exp}{Exp}
\DeclareMathOperator{\Poi}{Poi}

\newcommand{\R}{\mathbb{R}}

\newcommand{\ind}{\mathds {1}}

\renewcommand{\P}{\mathbb{P}}
\newcommand{\E}{\mathbb{E}}

\newcommand{\ve}{\varepsilon}

\newcommand{\tT}{\widetilde{T}}
\newcommand{\tX}{\widetilde{X}}

\newcommand{\ww}{\widehat{w}}
\newcommand{\wy}{\widehat{y}}

\newcommand{\tw}{\widetilde{w}}
\newcommand{\ty}{\widetilde{y}}

\begin{document}

\pagenumbering{arabic}
\title[Stochastic Spikes and Poisson Approximation of one-dim. SDEs]{Stochastic spikes and Poisson Approximation of one-dimensional stochastic differential equations with applications to continuously measured Quantum Systems }
\author[Kolb]{Martin Kolb}
\address{Institut f\"ur Mathematik, Universit\"at Paderborn, Warburger Str. 100,
        33098 Paderborn, Germany}
        \email{kolb@math.uni-paderborn.de}
\author[Liesenfeld]{Matthias Liesenfeld}
\address{Institut f\"ur Mathematik, Universit\"at Paderborn, Warburger Str. 100,
        33098 Paderborn, Germany}
        \email{liesenfe@math.uni-paderborn.de}
\maketitle
\begin{abstract}
Motivated by the recent contribution \cite{BB17} we study the scaling limit behavior of a class of one-dimensional stochastic differential equations which has a unique attracting point subject to a small additional repulsive perturbation. Problems of this type appear in the analysis of continuously monitored quantum systems. We extend the results of \cite{BB17} and prove a general result concerning the convergence to a homogeneous Poisson process using only classical probabilistic tools.
\end{abstract}

%%%%%%%%%%%%%%%%%%%%%%%%%%%%%%%%%%%%%%%%%%%%%%%%%%%%%%%%%%%%%%%%%%%%%%%%%%%%%%%%%%%%%
\section{Introduction}

Motivated by applications in Quantum Mechanics Bauer and Bernard investigated in the recent contribution \cite{BB17} scaling limits   $\lambda\rightarrow \infty$ and $\varepsilon\rightarrow 0$  for classes of stochastic differential equations of the form
\begin{align} \label{e:mainSDE}
dX_t=\frac{\lambda^2}{2}(\varepsilon \cdot b_1(X_t)-b_2(X_t))\,dt+\lambda \cdot \sigma(X_t)\,dB_t.
\end{align} 
More precisely, in case of constant $b_1>0$ and linear $b_2$ and $\sigma$, i.e. for stochastic differential equations of the form 
\begin{align} \label{e:BB-SDE}
dX_t=\frac{\lambda^2}{2}(\varepsilon - b X_t)\,dt+\lambda \cdot X_t \,dB_t
\end{align}
Bauer and Bernard rigorously study the non-trivial scaling limit of the process $(X_t)_{t\geq 0}$ in the regime $\lambda \rightarrow \infty$ and $\varepsilon \rightarrow 0$ such that $\lambda^2\varepsilon^{b+1}$ is constant and conjecture 
the validity of similar assertions for a larger class of stochastic differential equations of the type \eqref{e:mainSDE}. In this scaling limit the first hitting time of a level $z$ for the diffusion \eqref{e:BB-SDE} started at $x<z$ converges in distribution to a mixture of a point mass in zero and an exponential distributed random variable. Related questions for a slightly different model have previously been physically motivated and then analyzed by Bauer, Bernard and Tilloy in \cite{BBT15} and \cite{BBT16}. Observe that the diffusion given by \eqref{e:BB-SDE} is scale invariant, a fact which allows specific arguments and simplifies several calculations. Bauer and Bernard in particular proved that in the scaling limit $\lambda \rightarrow \infty$ and $\varepsilon \rightarrow 0$ with $\lambda^2\varepsilon^{b+1} = J$ constant the first hitting time of a level $z$ with start from $x<z$ converges in distribution to a convex combination of a exponential distributed random variable and the trivial random variable which is constant equal to zero. Using this result the authors also deduce a Poisson approximation for the number of hits above the level $z$. The analytic approach of Bauer and Bernard allows to cover also certain types of stochastic differential equations which are different from \eqref{e:BB-SDE} but still share the property of scale invariance. Using non-rigorous arguments the authors of \cite{BB17} come to the conjecture that the results will carry over to a larger rather general class of stochastic differential equations and they provide certain natural but not always precisely formulated conditions, under which the results are expected to hold. Our main aim is to provide a different rather elementary approach to the results of Bauer and Bernard, which allows to prove analogous results for general classes of stochastic differential equations, which do not necessarily satisfy a form of scale invariance. In particular we can extend the results to 'linearized version' of the stochastic differential equation describing the homodyne detection of Rabi oscillations. The resulting stochastic differential equation has a clear quantum mechanical background which  is in more detail described in \cite{BB17}. As a fact we will mainly rely on classical methods from probability theory such as Poisson approximation and some further mainly basic properties of diffusion processes. This is in contrast to the tools used by Bauer and Bernard which are analytic i.e. based on analysis of differential equations and basic It\^o theory for diffusions. Apart from extending the validity of the results to a larger class of stochastic differential equations we believe that our approach helps to put the results in a clear probabilistic perspective.

Let us stress that the results are related to known assertions about hitting times of large levels for diffusion processes such as e.g. \cite{M68} and \cite{BK98}. There the authors consider the behavior of hitting times of a high level and deduce that in an appropriate scaling limit this hitting time is exponentially distributed. We want to stress, that in the case of a non scale-invariant diffusion it does not seem possible to directly use known theorems concerning the extreme value behavior of hitting of large sets as given e.g. in \cite{M68} and \cite{BK98}. In the case of equation~\eqref{e:mainSDE} it is possible to connect the hitting of a fixed level $z$ when started from $\varepsilon$ into the question of hitting $z/\varepsilon$ with start in $1$. For this situation one can make direct use of the results in \cite{BK98} and of paragraph \,2, section \Romannum{5} in \cite{M68}. For start in a fixed point $x$ and for more non scale-invariant equations this does not seem possible. In any case due to the connections to the theory of quantum systems under continuous measurement we believe that our results and methods - which might not be that well known in the physics community - are of sufficiently broad interest and are useful in order to derive results for the most interesting higher dimensional situation.

The structure of the paper is the following: In Section 2 we introduce some essential notations and formulate the abstract version of our results, which are rigorously proved in section 3. This abstract result is based on a cycle decomposition of the diffusion and the usual renewal analysis of the associated renewal process. In section 4 we work through two classes of examples. The first class of stochastic differential equations we are dealing with are in some sense perturbation of equation~\eqref{e:BB-SDE}, which are still not covered by the results of Bauer and Bernard. The second fundamental example is deduced from the mathematical description of a 'linearized' version homodyne detection of Rabi oscillations.  

%%%%%%%%%%%%%%%%%%%%%%%%%%%%%%%%%%%%%%%%%%%%%%%%%%%%%%%%%%%%%%%%%%%%%%%%%%%%%%%%%%%%%
%%%%%%%%%%%%%%%%%%%%%%%%%%%%%%%%%%%%%%%%%%%%%%%%%%%%%%%%%%%%%%%%%%%%%%%%%%%%%%%%%%%%%
\section{Scaling limits of hitting times}
Let us give some basic definitions and notations. For $x>0$, we denote as $\P_x$ the probability measure for the diffusion process conditioned to start at $x$ and write $\E_x$ for the corresponding expectation. The hitting time for the process $(X_t)_{t\ge 0}$ of some level $z>0$ will be denoted as
\[T_z = \inf\lbrace t\ge 0\mid X_t = z\rbrace.\] 
We require the existence of two functions $0<\alpha(\ve) < \beta(\ve)$ for small $\ve>0$ which are differentiable in $0$ with $\lim_{\ve\downarrow 0} \beta(\ve) = \lim_{\ve\downarrow 0} \alpha(\ve)=0$. Let $(X_t^1)_{t\ge 0}$ denote the process $(X_t)_{t\geq 0}$ with $\lambda=1$ and $(\tX^1_t)_{t\ge 0}$ the process obtained from $(X_t^1)_{t\ge 0}$ by conditioning on $\lbrace T_{\alpha(\ve)} < T_z \rbrace$ via a $h$\=/transform in the sense of Doob (see e.g. \cite{P95}, chapter 4, section 1). We introduce the following quantities:
\begin{align*}
	\widetilde{\sigma}_0=0,\,\widetilde{\tau}_1=\inf\lbrace t\geq 0\mid X_t^1=\beta(\ve)\rbrace,\,\widetilde{\sigma}_1=\inf\lbrace t\geq \widetilde{\tau}_1\mid \tX_t^1=\alpha(\ve)\rbrace;
\end{align*} furthermore, for $i\ge 2$:
\begin{align*}
\widetilde{\tau}_i=\inf\lbrace t\geq \widetilde{\sigma}_{i-1}\mid X_t^1=\beta(\ve)\rbrace,\,\widetilde{\sigma}_i=\inf\lbrace t\geq \widetilde{\tau}_{i}\mid \tX_t^1=\alpha(\ve)\rbrace.
\end{align*}
\begin{figure}
\includegraphics{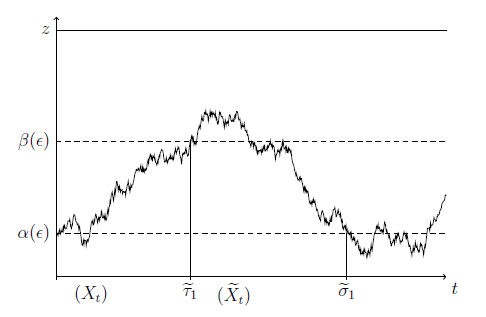}
\caption{Illustration of the cycle decomposition given by the stopping times $\widetilde{\tau}_i$ and $\widetilde{\sigma}_i$}
\end{figure}
Starting at $\alpha(\ve)$ we run the process $(X_t)_t$ until it hits $\beta(\ve)$ (observe that the conditioning event $\lbrace T_{\alpha(\ve)} < T_z \rbrace$ has full probability as we start in $\alpha(\ve)$), then we run the conditioned process $(\tX_t)_{t\geq 0}$ starting in $\beta(\ve)$ until we hit $\alpha(\ve)$. 

When started at $\beta(\ve)$, the probability to hit $z$ before $\alpha(\ve)$ will be denoted as 
\[p_{\ve,z} := \P_{\beta(\ve)}(T_z < T_{\alpha(\ve)}).\] 
If a cycle means a piece of the diffusion path starting at $\alpha(\ve)$, moving to $\beta(\ve)$ and then returning to $\alpha(\ve)$ then $p_{\ve,z}$ describes the probability that the cycle was completed without hitting $z$. 

By (generalized) scaling limit we will mean the limiting process as $\lambda\to\infty$ and $\ve\downarrow 0$ along the curve $\lambda^2 p_{\ve,z} = J >0$. Especially, for the generalized scaling limit to be defined, it is required, that $\ve\downarrow 0$ implies $p_{\ve,z}\to 0$. Let us now employ the following standing assumptions on the considered stochastic differential equation~(\ref{e:mainSDE}): \begin{enumerate}[label=(A\theenumi)] 
	\item There exists a (weak) solution to the SDE~(\ref{e:mainSDE}) in the sense of Definition~25.1 in \cite{B11}, which is unique in law.
	\item The expected cycle length converges to some positive real number independent of $z$: \[\E_{\alpha(\ve)}[\widetilde{\sigma}_1] \xrightarrow[\ve\downarrow 0]{} \kappa^{-1} \in (0,\infty). \]
	\item For small $\ve>0$, the cycles have finite second moment uniformly in $\ve$: \[\limsup_{\ve\downarrow 0} \E_{\alpha(\ve)} [\widetilde{\sigma}_1^2] < \infty. \]
\end{enumerate}
\begin{remark}
	Technically, (A3) may be weakened by $\limsup_{\ve\downarrow 0} \E_{\alpha(\ve)}[\widetilde\sigma_1^{1+\rho}] < \infty$ for some positive $\rho>0$. That is, only $(1+\rho)$-th moment is actually needed. The conditions (A1), (A2) and (A3) are rather natural and not too restrictive. 
\end{remark}

With the help of a regeneration structure based on cycle decompositions we will show
\begin{proposition} \label{p:mainProp}
	Assume that (A1), (A2) and (A3) are satisfied, then in the scaling limit $\lambda \rightarrow \infty$, $\varepsilon \rightarrow 0$ with $\lambda^2p_{\varepsilon,z}=J   \in (0,\infty)$ we have 
	\begin{displaymath}
	\lim_{scaling}\P_{\alpha(\ve)}\bigl(T_z>T)= e^{-\kappa J T}.
	\end{displaymath}
\end{proposition}
This result gives the almost exponential behavior of the hitting of a fixed level $z$, when started very close to zero. In order to deduce the result when started from a fixed level $0<x<z$ we assume the following conditions: 
\begin{enumerate}[label=(B\theenumi)]
	\item In the generalized scaling limit for any $z>0$ and $0<x<z$, under $\P_x$, \[T_{\alpha(\ve)} \wedge T_z \xrightarrow[scaling]{\mathcal{D}} 0, \]
	      and the law of $T_{\alpha(\ve)}$ under $\P_x(\cdot \mid T_{\alpha(\ve)} < T_z)$ converges to the point mass in zero for any $z>0$.
	\item Furthermore, the limit \[\P_x (T_{\alpha(\ve)} < T_z) \xrightarrow[\ve\to 0]{} \alpha_{x,z} \in (0,1) \]%\frac{
		%\int_x^z \exp \left(-\int_y^z \frac{b_2(l)}{\sigma^2(l)} \,dl\right) \,dy}{\int_0^z \exp \left(-\int_y^z \frac{ b_2(l)}{\sigma^2(l)} \,dl\right) \,dy},\]
		exists for all $z>0$, $0<x<z$.
\end{enumerate} 
%\begin{remark}
%	Instead of the first assumption in (B1), it would suffice to require the convergence in probability, only. In the considered example class of SDEs we will prove the proposed $L^1$ limit.
%\end{remark}
\begin{remark}
	Our assumptions (A1) -- (A3) and (B1) -- (B2) are natural and related but not fully comparable to the conditions formulated by Bauer and Bernard. We point at some similarities. Condition \romannum{9}) in \cite{BB17} essentially corresponds to (A3) and the assumption $p_{\ve,z}\to 0$ is related to \romannum{2}). \romannum{1}) is encoded in the example below as (E2) and (E3).
\end{remark}
We are now ready to state our main result. 
\begin{theorem}\label{t:mainThm}
Assume that the conditions (A1) to (A3), (B1) and (B2) are satisfied, then in the scaling limit $\lambda \rightarrow \infty$, $\varepsilon \rightarrow 0$ with $\lambda^2p_{\varepsilon,z}=J   \in (0,\infty)$ the law of the hitting time $T_z$ when started at $0<x<z$ equals 
\[	
    (1-\alpha_{x,z}) \,\delta_0 + 
	\alpha_{x,z}  \Exp_{J\kappa}.
\]
The choice of the level one in $p_{\varepsilon,1}$ and in the definition of $q(z)$ respectively is of course rather arbitrary. 
%where 
%\begin{displaymath}
%\alpha_x:=\frac{\int_0^x \exp \left(-\int_y^z \frac{b_2(l)}{\sigma^2(l)} \,dl\right) \,dy}{\int_0^z \exp \left(-\int_y^z %\frac{ b_2(l)}{\sigma^2(l)} \,dl\right) \,dy}.
%\end{displaymath}
\end{theorem}
\begin{remark}
	In the special case of equation~\eqref{e:BB-SDE} this result corresponds to Corollary~3 in \cite{BB17}.
\end{remark}
Theorem~\ref{t:mainThm} can be interpreted in the following way, which has also been observed in \cite{BB17}. If the diffusion process starts at the point $x$ and wants to reach level $z$ then there are two options: Either the process reaches level $z$ without coming close to zero and in the scaling limit this takes no time or it first reaches a neighborhood of zero. Once it has reached the neighborhood of $0$ it needs many trials to get up to level $z$ and each trial has low success probability (see e.g. \cite{B90}). The latter follows from the form of the stochastic differential equation; the drift is weak near zero and the diffusion is slowed down near zero. The proof of this result will exactly follow this picture and we will make this rigorous in the following section. \\

The asymptotic of $p_{\ve,z} \to 0$, as $\ve\downarrow 0$ may depend on the value of $z$. We want to consider the $z$-free scaling limit. For $z\in(0,\infty)$ define \[
q(z) := \begin{cases}
1-\alpha_{z,1}, & z\le 1,\\
(1-\alpha_{1,z})^{-1} ,& z>1.
\end{cases}
\]

\begin{corollary}
	Assume that all conditions (A1) to (A3), (B1) and (B2) are satisfied, then in the scaling limit $\lambda \rightarrow \infty$, $\varepsilon \rightarrow 0$ with $\lambda^2p_{\varepsilon,1}=J   \in (0,\infty)$ the law of the hitting time $T_z$ when started at $0<x<z$ equals 
	\[	
	(1-\alpha_{x,z}) \,\delta_0 + 
	\alpha_{x,z}  \Exp_{\kappa J/q(z)}.
	\]
\end{corollary}
\begin{proof}
	\begin{align*}
		\lim\limits_{scaling} \P_x(T_z>T) &= \lim\limits_{scaling} \P_x(T_z(X^1) > \lambda^2T) \\
		&= \lim\limits_{\ve\downarrow 0} \P_x\left(T_z(X^1)\frac{p_{\ve,z}}{JT} > q^{-1}(z) + \left(\frac{p_{\ve,z}}{p_{\ve,1}}-q^{-1}(z)\right) \right)\\
		&= \lim\limits_{\substack{\lambda\to\infty\\\ve\to 0\\ \lambda^2p_{\ve,z} = J}} \P_x\left(\frac{T_z}{T}> q^{-1}(z) +    \left(\frac{p_{\ve,z}}{p_{\ve,1}}-q^{-1}(z)\right) \right).
	\end{align*} 
Below we will show, that the convergence 
\begin{align}\label{e:claimLim}
\lim\limits_{\ve\downarrow 0}\frac{p_{\ve,z}}{p_{\ve,1}} = \frac{1}{q(z)} 
\end{align} 
holds true. Using Theorem~\ref{t:mainThm} this shows 
  \[
	\limsup_{scaling} \P_x(T_z>T) \le \alpha_{x,z} e^{-\kappa J T (q^{-1}(z)-\delta)}
	\]  
	and 
	\[
	\liminf_{scaling} \P_x(T_z>T) \ge \alpha_{x,z} e^{-\kappa J T (q^{-1}(z)+\delta)}
	\] 
	for $\delta>0$ arbitrary, hence implying the assertion.
	
It remains to prove the claim~\eqref{e:claimLim}. By the strong Markov property 
\[
	p_{\ve,1\vee z} = p_{\ve, 1\wedge z} \cdot \P_{1\wedge z}(T_{1\vee z} < T_{\alpha(\ve)}).
\] 
It follows from condition (B2) 
\[
\frac{p_{\ve,z}}{p_{\ve,1}} = 
\begin{cases}
	(1-\P_z(T_{\alpha(\ve)} < T_1))^{-1}, & z\le 1,\\
	1-\P_1(T_{\alpha(\ve)} < T_z) ,& z>1
\end{cases} 
	\xrightarrow[\ve\downarrow 0]{} q^{-1}(z). 
\]
This gives the required assertion. 
\end{proof}
\begin{remark}
In both examples worked out below the scaling limit relation $\lambda^2 p_{\varepsilon,1} = const$ is essentially (meaning up to some arbitrary positive multiplicative constant) equivalent to choosing the curve $\lambda^2 Z_{\ve}=const$ which is used \cite{BB17} in order to formulate the general conjecture. There, \[Z_\ve := \int_0^\infty \frac{1}{x^4}\exp\left( -\frac{\ve}{3} \frac{1}{x^3}  + \frac{b}{2}\frac{1}{x^2} \right)\, dx.\] denotes the total mass of some invariant measure, cf. condition \romannum{7}) in section 6.1 main conjectures in \cite{BB17}. Also, $q(z)$ in that article is the same as our $q(z)$ here if the limit in (B2) has the form as in the examples. Note, that our main result corresponds to Conjecture B (\romannum{1}) and (\romannum{2}).
\end{remark}
 
\section{An embedded approximate Poisson process}

In \cite{BB17}, the distribution of the first hitting time $T_z$ is deduced by calculating the Laplace transform of $T_z$, i.e. the expectation
\begin{displaymath}
\E_x\bigl[e^{-s T_z}\bigr],\quad s >0,\,0<x<z
\end{displaymath}
making use of the fact that they solve certain ordinary differential equations.  Our approach has a somewhat different more probabilistic flavor. We are using the following rather classical strategy:
\begin{itemize}
	\item Starting the diffusion near zero, we introduce stopping times, which decompose the path up to an arbitrary time $T$ into cycles.
	\item During every cycle, the diffusion reaches with a small probability the level $z$.
	\item Counting only the hits of level $z$ now up to a time $\lambda^2T$ results in an approximate Poisson process. 
\end{itemize}
As mentioned above we call cycle a path from ${\alpha(\ve)}$ to ${\beta(\ve)}$ and back to ${\alpha(\ve)}$ when $\lambda$ is set to equal $1$. If we now speed up the time scale which is done by introducing the large time scale factor $\lambda$ we have many cycles in a time interval $[0,T]$ and in each cycle we hit the level $z$ with small probability. This is the standard situation, where the Poisson heuristic should apply. 

\subsection{A Thinned Renewal Process}
In order to motivate our approach we define the counting variable
\begin{displaymath}
{N}(T)=\max \lbrace i \in \mathbb{N}_0 \mid  {\sigma}_i \leq T\rbrace,
\end{displaymath} where \begin{equation*}
{\sigma}_0=0,\,{\tau}_1=\inf\lbrace t\geq 0\mid X_t^1=\beta(\ve)\rbrace,\,{\sigma}_1=\inf\lbrace t\geq {\tau}_1\mid X_t^1={\alpha(\ve)}\rbrace,
\end{equation*}
\begin{equation*}
{\tau}_i=\inf\lbrace t\geq {\sigma}_{i-1}\mid X_t^1=\beta(\ve)\rbrace,\,{\sigma}_i=\inf\lbrace t\geq {\tau}_{i}\mid X_t^1={\alpha(\ve)}\rbrace.
\end{equation*}
The quantity ${N}(T)$ encodes the number of cycles completed up to time $T$ and we will use known results from renewal theory (see e.g. \cite{D05}, \cite{L76} and \cite{Y91}).  For given $z>0$ we are actually not interested in the number of completed cycles up to time $T$ but in the number of cycles up to time $T$, which do cross the level $z$. Thus we have to delete those cycles which do not cross the level $z$ and we observe that this happens with probability $1-p_{\varepsilon,z}$.\\
\\
As a first motivation we consider the thinned-rescaled point process obtained by retaining every point of ${N}$ with probability $p_{\ve,z}$ independently of the other points and of the point process ${N}(T)$, and then replacing the retained point at time instant $t_i$ by a point at $\lambda^{-2}\cdot t_i$. Let us denote this counting process by 
\begin{displaymath}
{\mathfrak{N}}_{p_{\ve,z},\lambda}(T),\quad T\geq 0.
\end{displaymath}
We observe that 
\begin{equation*}
\begin{split}
{\mathfrak{N}}_{p_{\ve,z},\lambda}(T) = \xi_1+\dots +\xi_{{N}(\lambda^2T)},
\end{split}
\end{equation*}
where the random variables $\xi_1,\xi_2,\dots$ are independent and identically distributed with $\mathrm{P}(\xi_i=1)=p_{\varepsilon,z},\,\mathrm{P}(\xi_i=0)=1-p_{\varepsilon,z}$ and independent of the process $({N}(t))_{t\geq 0}$. This thinned counting process in fact converges to a Poisson process as can be deduced using standard results in the literature. Obviously, the independent thinning does not precisely describe what we are really interested in. 

\subsection{Poisson Limits in the high noise regime}
We thus consider the probability, that none of the cycles of our original process up to time $\lambda^2T$ has reached level $z$, i.e. we investigate
\begin{align*}
&\sum_{k=0}^{\infty}\P_{\alpha(\ve)}\left({N}(\lambda^2T)=k,\forall 1\leq i\leq k:\sup_{\tau_{i}\leq t\leq \sigma_i}X_t^1 < z \right)\\
=&\sum_{k=0}^{\infty} \Bigl[ \P_{\beta(\ve)}\bigl(T_{\alpha(\ve)}<T_z\bigr)^k  \\
&\qquad \times \P_{\alpha(\ve)}\Bigl({N}(\lambda^2T)=k, \forall 1\leq i\leq k:\sup_{\tau_{i}\leq t\leq \sigma_i}X_t^1 < z\Bigr)  \, \P_{\beta(\ve)}\bigl(T_{\alpha(\ve)}<T_z\bigr)^{-k} \Bigr].
\end{align*}
Let us observe that using results on the relation between conditioning and $h$\=/transforms we have for $k\ge 1$
\begin{align}\begin{split}\label{e:NandNTilde}
&\P_{\alpha(\ve)}\biggl({N}(\lambda^2T)=k,\forall 1\leq i\leq k: \sup_{\tau_{i}\leq t\leq \sigma_i}X_t^1 < z\biggr) \, \P_{\beta(\ve)}\bigl(T_{\alpha(\ve)}<T_z\bigr)^{-k}\\
&=\P_{\alpha(\ve)}\bigl(\widetilde{N}(\lambda^2T)=k\bigr),
\end{split}\end{align}
where the process $(\widetilde{N}(T))_{T\geq 0}$ is the counting process $\widetilde{N}(T)=\max\lbrace n\mid \widetilde{\sigma}_n<T\rbrace$. We need to stress that the involved quantities depend on $\varepsilon$ even though the notation does not make this explicit. 
%\begin{lemma}\label{l:NTildeProb}
%For every $k\in \mathbb{N}$ one has 
%\begin{equation*}
%\P\biggl({N}(\lambda^2T)=k,\forall 1\leq i\leq k: \sup_{\tau_{i}\leq t\leq \sigma_i}X_t^1\leq z\biggr)\P_{2\varepsilon}\bigl(T_{\varepsilon}^1<T_z^1\bigr)^{-k}=\P\bigl(\tilde{N}_\ve(\lambda^2 T)=k\bigr).
%\end{equation*}
%\end{lemma}

\begin{proposition} \label{p:LimTillCycle}
Under (A1) to (A3) to hold, we have 
\[
\lim_{scaling}\sum_{k=0}^{\infty}\P_{\alpha(\ve)}\biggl({N}(\lambda^2T)=k,\forall 1\leq i\leq k : \sup_{\tau_{i}\leq t\leq \sigma_i}X_t^1 < z\biggr)= e^{-\kappa JT}.
\]
\end{proposition}
\begin{proof}
Defining \[\widetilde{\mathfrak{N}}_{p_{\varepsilon,z},\lambda}(T):=\widetilde{\xi}_1+\dots +\widetilde{\xi}_{\widetilde{N}(\lambda^2T)}\] with $(\widetilde{\xi}_i)_{i\ge 1}$ being an independent family of Bernoulli distributed random variables with $\mathrm{P}(\widetilde{\xi}_1=1)= p_{\varepsilon,z}$ and independent of the counting process $\widetilde{N}(T)$ we notice using equation~(\ref{e:NandNTilde}) 
\begin{align*}
	&\sum_{k=0}^{\infty}\P_{\alpha(\ve)}\biggl({N}(\lambda^2T)=k,\forall 1\leq i\leq k : \sup_{\tau_{i}\leq t\leq \sigma_i}X_t^1 < z\biggr) \\
	&= \sum_{k=0}^{\infty} \P_{\beta(\ve)}(T_{\alpha(\ve)} < T_z)^k \, \P_{\alpha(\ve)}(\widetilde{N}(\lambda^2T)=k) = \mathrm{P}(\widetilde{\mathfrak{N}}_{p_{\ve,z},\lambda}(T) = 0).
\end{align*}
We observe that by standard results on Poisson approximation (see e.g. equation~(23) in \cite{Y91}) for every $T>0$
\begin{equation*}
\begin{split}
d_{TV}(\widetilde{\mathfrak{N}}_{p_{\varepsilon,z},\lambda}(T),\Poi_{{\kappa} JT})\leq \frac{p_{\varepsilon,z}}{2\sqrt{1-p_{\varepsilon,z}}}+\E_{\alpha(\ve)}\bigl[\bigl| p_{\varepsilon,z}\widetilde{N}(\lambda^2T)-{\kappa}JT\bigr|\bigr].
\end{split}
\end{equation*}
Therefore it is sufficient to show the convergence 
\begin{equation*}%\label{e:uniformlaw}
\lim\limits_{scaling} \E_{\alpha(\ve)}\bigl[\bigl| p_{\varepsilon,z}\widetilde{N}(\lambda^2T)-{\kappa}JT\bigr|\bigr] =0.
\end{equation*}
%Due to (A2) and (A3) we can apply a suitable version of the uniform renewal theorem such as Theorem 10 in \cite{L76} in order to conclude \eqref{e:uniformlaw} and therefore the assertion of the theorem: \\
With $\kappa_\ve := 1/\E_{\alpha(\ve)}[\widetilde{\sigma}_1]$ we obtain 
\begin{align*}
&\E_{\alpha(\ve)}\bigl[\bigl| p_{\varepsilon,z}\widetilde{N}(\lambda^2T)-{\kappa}JT\bigr|\bigr] 
= \E_{\alpha(\ve)}\bigl[\bigl| \lambda ^{-2} \widetilde{N}(T\lambda^2) -\kappa T \bigr|\bigr] \\
&\le \E_{\alpha(\ve)}\bigl[\bigl| \lambda ^{-2} \widetilde{N}(\lambda^2T) -\kappa_\ve T \bigr|\bigr] + |\kappa-\kappa_\ve|\cdot T.
\end{align*} 
The vanishing of $|\kappa -\kappa_\ve|\to 0$ is a reformulation of (A2) and due to (A2) together with (A3) we can apply a suitable version of the uniform renewal theorem such as Theorem~10 in \cite{L76} in order to conclude \[
\lim_{scaling} \E_{\alpha(\ve)}\bigl[\bigl| \lambda ^{-2} \widetilde{N}(T\lambda^2) -\kappa_\ve T \bigr|\bigr] \le \lim_{\lambda\to\infty} \sup_{\frac{z}{4}>\ve>0} \E_{\alpha(\ve)}\bigl[\bigl| \lambda ^{-2} \widetilde{N}(T\lambda^2) -\kappa_\ve T \bigr|\bigr]=0.
\] This finishes the proof.
\end{proof}

Now we only have to do one more last step. Observe that we have not yet reached exactly what we want. In order to describe the event $\lbrace T_z>T\rbrace$ we need to consider the event 
\begin{displaymath}
\lbrace {N}(\lambda^2T)=k,\forall {\sigma}_k\leq t\leq T:X_t^1<z\rbrace,
\end{displaymath}
this means we also have to make sure that during the cycle started before time $T$ but not completed before this time the level $z$ has not been hit.
\begin{proof}[Proof of Proposition~\ref{p:mainProp}]
%In the scaling limit $\lambda \rightarrow \infty$, $\varepsilon \rightarrow 0$ with $\lambda^2p_{\varepsilon,z}=J   \in (0,\infty)$ we have 
%\begin{displaymath}
%\lim_{scaling}\P_\varepsilon\bigl(T_z>T)= e^{-\kappa J T}.
%\end{displaymath}

%\begin{proof}
From the fact $T \in [\sigma_{N(\lambda^2T)},\sigma_{N(\lambda^2T)+1})$ we see \[
\P_{\alpha(\ve)}(T_z>T)  \in \left(\P_{\alpha(\ve)}(T_z>\sigma_{N(\lambda^2T)+1}),\P_{\alpha(\ve)}(T_z>\sigma_{N(\lambda^2T)})\right]
\]
and by the previous Proposition~\ref{p:LimTillCycle} the upper bound has the asserted scaling limit. For the lower, we may define \[\widetilde{\mathfrak{N}}_{p_{\varepsilon,z},\lambda}^+(T):=\widetilde{\xi}_1^++\dots +\widetilde{\xi}_{\widetilde{N}(\lambda^2T)+1}^+\] with $(\widetilde{\xi}_i^+)_{i\ge 1}$ being an independent family of Bernoulli distributed random variables with $\mathrm{P}(\widetilde{\xi}_1^+=1)= p_{\varepsilon,z}$ and independent of the counting process $\widetilde{N}(T)$ and repeat the argumentation as in the proof of Proposition~\ref{p:LimTillCycle}:
\begin{align*}
&\sum_{k=0}^{\infty}\P_{\alpha(\ve)}\biggl({N}(\lambda^2T)=k,\forall 1\leq i\leq k+1 : \sup_{\tau_{i}\leq t\leq \sigma_i}X_t^1 < z\biggr) \\
&= \sum_{k=0}^{\infty} \P_{\beta(\ve)}(T_\ve < T_z)^{k+1} \, \P_{\alpha(\ve)}(\widetilde{N}(\lambda^2T)=k) = \mathrm{P}(\widetilde{\mathfrak{N}}_{p_{\ve,z},\lambda}^+(T) = 0).
\end{align*} Then, in the scaling limit \[
\lim\limits_{scaling} \E_{\alpha(\ve)}\bigl[\bigl| p_{\varepsilon,z}\bigl(\widetilde{N}(\lambda^2T)+1\bigr)-{\kappa}JT\bigr|\bigr] =0
\] still holds and the assertion is shown.

\end{proof}

Let us now start from a point $x>0$ and derive the law of $T_z$ with respect to $\P_x$. Starting at $x$ there are two cases to consider:
\begin{itemize}
\item The diffusion reaches ${\alpha(\ve)}$ before hitting $z$.
\item The process hits $z$ before visiting ${\alpha(\ve)}$. 
\end{itemize}
%\begin{theorem}
%$T_z$ converges in the scaling limit under $\P_x$, $x<z$ to \[\frac{x^{b+1}}{z^{b+1}}\delta_0 + \left(1-\frac{x^{b+1}}{z^{b+1}}\right)\Exp_{J\kappa}.\]
%\end{theorem}
\begin{proof}[Proof of Theorem~\ref{t:mainThm}]
\[
\P_x(T_z>T) = \P_x(T_{\alpha(\ve)}\wedge T_z > T) + \P_x(T_{\alpha(\ve)} \le T < T_z).
\]
From (B1) it follows, that the first summand vanishes in the scaling limit and writing 
\[
\P_x(T_{\alpha(\ve)} \le T < T_z) = \P_x(T_{\alpha(\ve)}< T_z)\cdot \P_x (T_{\alpha(\ve)} \le T < T_z \mid T_{\alpha(\ve)} < T_z)
\] 
we see by (B2) that the first factor of that product has the scaling limit 
\[
\lim_{scaling}\P_x(T_{\alpha(\ve)}< T_z) =\frac{\int_x^z \exp \left(-\int_y^z \frac{b_2(l)}{\sigma^2(l)} \,dl\right) \,dy}{\int_0^z \exp \left(-\int_y^z \frac{ b_2(l)}{\sigma^2(l)} \,dl\right) \,dy}.
\]
For the second factor, with notation $\tT_z := T_z(\tX)$ an application of the strong Markov property at time $\tT_{\alpha(\ve)}$ leads to  
\begin{align*}
	&\P_x (T_{\alpha(\ve)} \le T < T_z \mid T_{\alpha(\ve)} < T_z) =\P_x(\tT_{\alpha(\ve)} \le  T < \tT_z) \\
	&=\int\limits_{\lbrace \tT_{\alpha(\ve)} \le T\rbrace} \P_{\alpha(\ve)} (T_z > T - \tT_{\alpha(\ve)}(\omega)) \,\P_x(d\omega)\\
	& =\int\limits_{\lbrace \tT_{\alpha(\ve)} \le T\rbrace} \left(1-\P_{{\alpha(\ve)}} (T_z \le T - \tT_{\alpha(\ve)}(\omega)) \right) \,\P_x(d\omega)\\
	&= \P_x (\tT_{\alpha(\ve)} \le T)-\int\P_{\alpha(\ve)} (\tT_{\alpha(\ve)}(\omega) + T_z \le T  ) \,\P_x(d\omega).
\end{align*}

The first summand has scaling limit $1$ and the integral may be seen as a probability of the convolution \begin{align*}
&\int\P_{\alpha(\ve)} (\tT_{\alpha(\ve)}(\omega) + T_z \le T  ) \,\P_x(d\omega) = \left[ \left(\P_x \circ(\tT_{\alpha(\ve)})^{-1} \right) \ast \left(\P_{\alpha(\ve)} \circ (T_z)^{-1}\right) \right] ([0,T]).
\end{align*} Due to the independence the characteristic function (as mapping of $s$) is the product \[\E_x [e^{is\tT_{\alpha(\ve)}}] \cdot \E_{\alpha(\ve)}[e^{isT_z}]\] and while we see the first factor has scaling limit $1$ we finish the proof by recalling Proposition~\ref{p:mainProp}.
\end{proof}

\section{Examples}
%Of course one aims to find conditions under which the above results extend to equations of the form
%\begin{equation}
%dX_t=\frac{\lambda^2}{2}\bigl(\varepsilon b_1(X_t)-b_2(X_t)\bigr)\,dt+\lambda \sigma(X_t)\,dB_t.
%\end{equation}

In the section we present two important classes of examples which illustrate our approach. The second example is motivated by a specific quantum mechanic situation.

\begin{remark}
	The formal generator associated to our SDE is given by
	\begin{equation*}
	L:= \frac{\lambda^2}{2}\sigma^2(x)\frac{d^2}{dx^2}+\frac{\lambda^2}{2}\bigl(\varepsilon b_1(x)-b_2(x)\bigr)\frac{d}{dx}.
	\end{equation*}
	The scale function $s$ (up do multiplicative constants) defined by the relation 
	\[
	\P_x(T_R < T_r) = \frac{s(x)-s(r)}{s(R)-s(r)}
	\] 
	for $0<r<x<R$ is given by 
	\[
	s(x) = \int_c^x \exp \left(-\int_c^y \frac{\ve b_1(l) -b_2(l)}{\sigma^2(l)} \,dl\right) \,dy = \int_c^x 1/p_c(y) \,dy.
	\] 
	The speed measure is 
	\begin{align*}
	m(dx) = \frac{2}{\lambda^2 \sigma^2(x)s'(x)} \,dx = \frac{2p(x)}{\lambda^2 \sigma^2(x)}  \, dx = 2r(x) \, dx.
	\end{align*} 
	The generator can be written in divergence form as 
	\begin{equation}\label{e:div-form}
	Lu(x) = \frac{1}{2r(x)} \frac{d}{dx} \left(p(x) \frac{du}{dx} (x)\right).
	\end{equation}
	For more details we refer to standard book on stochastic processes such as e.g. \cite{B11}.
\end{remark}

\subsection{Asymptotic linear stochastic differential equations}
 
\begin{enumerate}[label=(E\theenumi)]
	\item Let $b_1$ a positive continuously differentiable mapping from $[0,\infty)$ to $(0,\infty)$ uniformly bounded away from $0$ and from above, i.e. \[0<  a_- := \inf b_1,\quad \sup b_1 =: a_+ < \infty;\] specifically, $a:= b_1(0) >0$.
	\item $b_2$ a nonnegative twice continuously differentiable function on $[0,\infty)$ with $b_2(0)=0$ and $b:= b_2'(0) > 0$.
	\item $\sigma: [0,\infty) \to [0,\infty)$ twice continuously differentiable with \begin{enumerate}[leftmargin=1.0cm] 
		\item $\sigma(x) = 0 \Leftrightarrow x=0$,
		\item $\sigma := \sigma'(0) > 0$.
	\end{enumerate}
\end{enumerate}

This example class can be viewed as generalization of the specification 
\begin{equation}\label{e:scaleinv}
b_1(x) := 1,\qquad b_2(x) := b\cdot x, \qquad \sigma(x) := x
\end{equation}
in the sense that at the origin the coefficients exhibit the same behavior. Note, that in the situation of \eqref{e:scaleinv} a strong form of scale invariance holds, i.e. $Y_t := X_t/\ve$ fulfills the SDE 
\[
dY_t = \frac{\lambda^2}{2}\left(1-b\cdot Y_t\right) \, dt + \lambda \cdot Y_t \, dB_t 
\] 
making it plausible to choose $\alpha(\ve) $ and $\beta(\ve) $ of linear order. Our next goal is to perform the needed calculations for showing (A2), (A3), (B1) and (B2) where we set $\alpha(\ve):=\ve$ and $\beta(\ve) := 2\ve$.

\begin{remark}  \label{r:taylor}
	By Taylor's theorem, there is $M>0$, $(a\wedge b \wedge \sigma^2)/(2M)>\delta_0>0$ such that 
	\begin{align*}
		|b_1(x)-a| \le M x,\,	|b_2(x)-bx| \le M x^2,\, |\sigma^2(x)-\sigma^2x^2| \le M x^3
	\end{align*}
	for all $0\le x \le \delta_0$.
\end{remark}

To verify $p_{\ve,z}\xrightarrow{\ve\downarrow 0} 0$ we choose $\delta_0>0$ so that the inequalities in Remark~\ref{r:taylor} above hold on $x\le \delta_0$, set $\delta:= \delta_0 \wedge z/2$ and write 
\begin{align}\label{e:p_ez} 
p_{\ve,z}:= \P_{2\ve}(T_z<T_\ve) = \frac{\int_\ve^{2\ve} 1/p_\delta(y) \, dy}{\int_\ve^{z} 1/p_\delta(y) \, dy}.
\end{align}
Then the numerator tends to $0$ as 
\begin{align} \label{e:p_ezNum}
\int_\ve^{2\ve} 1/p_\delta(y) \, dy \le  \int_\ve^{2\ve} \exp\left(3\,\ve\, a/\sigma^2 \cdot  (1/y-1/\delta)\right) (y/\delta)^{b/(3\sigma^2)} \, dy \to 0,
\end{align} 
whereas the denominator does not vanish:
\begin{align} \label{e:p_ezDen}
	\int_\ve^{z} 1/p_\delta(y) \, dy 
	&\ge \int_\delta^{z} \exp\left(-\delta \int_\delta^y \frac{b_1(l)}{\sigma^2(l)}\, dl\right) \exp\left(\int_\delta^y \frac{b_2(l)}{\sigma^2(l)} \,dl\right)\, dy >0.
\end{align}\\

In order to prove the validity of (A2) we investigate $\E_\ve[\widetilde{\sigma}_1]$ for small $\ve>0$. As preparation and for later use we college some explicit estimates 
\begin{lemma} \label{l:intEst}
	The following assertions are true:
	\begin{itemize}
	\item[a)] For $0<y\le w< \delta_0/\ve$ the estimates 
	\begin{align*}
	\ve^2\frac{r(y\ve)}{p(w\ve)} \,  \begin{matrix}
	\ge\\
	\le
	\end{matrix} \, & \frac{1}{\sigma^2 \cdot y^2 \pm My^3\ve}\left( \frac{\sigma^2 \mp Mw\ve}{\sigma^2 \mp My\ve} \cdot \frac{y}{w}\right)^{ \pm \ve M(a+\sigma^2)/\sigma^4}    \\     
	&\exp\left( \frac{a}{\sigma^2} (1/w - 1/y) \right) 
	 \left(\frac{w}{y}\right)^{b/\sigma^2} \left(\frac{ \sigma^2 \pm M w\ve}{\sigma^2 \pm My\ve}\right) ^{\mp (b/\sigma^2+1)}  
	\end{align*} 
	\item[b)] For $1 < w\le y < \delta/\ve $ with $\delta:= \delta_0 \wedge z/2$ with $\pm$ and $\mp$ interchanged except the first $\pm$ in the denominator.
	\item [c)] For $0<y,w<\delta_0/\ve$ and for $1 < w\le y < \delta/\ve $ we have
	\begin{displaymath}
	\lim_{\ve \rightarrow 0}\ve^2\frac{r(y\ve)}{p(w\ve)}=\frac{1}{\sigma^2}e^{\frac{a}{\sigma^2}(1/w-1/y)}\frac{w^{b/\sigma^2}}{y^{b/\sigma^2+2}}.
	\end{displaymath}
	\end{itemize}
\end{lemma} 
\begin{proof}
	In order to prove assertion a) we use Remark~\ref{r:taylor} in the case $y\le w$ and conclude 
	\begin{align*}
	\frac{r(y\ve)}{p(w\ve)} &= \frac{1}{\sigma^2(y\ve)}\exp\left(\ve \int_{w\ve}^{y\ve} \frac{b_1(l)}{\sigma^2(l)} \,dl  \right) \exp\left( \int_{y\ve}^{w\ve} \frac{b_2(l)}{\sigma^2(l)} \,dl \right) \\
		& \,  \begin{matrix}
		\ge\\
		\le
		\end{matrix} \,  \frac{1}{\sigma^2 \cdot (y\ve)^2 \pm M(y\ve)^3}\exp\left(\ve \int_{w\ve}^{y\ve} \frac{a\pm Ml}{\sigma^2 l^2 \mp Ml^3} \,dl  \right)  \exp\left( \int_{y\ve}^{w\ve} \frac{b\mp Ml}{\sigma^2 l \pm Ml^2} \,dl \right).
	\end{align*} 
	An application of partial fraction decomposition allows the integrals explicitly and give the estimates given in assertion a). 
	The proof of b) is completely analogous and assertion c) follows immediately from a) and b). 
\end{proof}

Since in (A2) we consider the process $X_t = X_t^1$ with parameter $\lambda=1$, we now write $T_z$ for $T_z(X^1)$. By the strong Markov property, \[\E_{\alpha(\ve)}[\widetilde{\sigma}_1] = \E_{\alpha(\ve)}[T_{\beta(\ve)}] + \E_{\beta(\ve)}[T_{\alpha(\ve)} \mid T_{\alpha(\ve)} < T_z] = \E_{\alpha(\ve)}[T_{\beta(\ve)}] + \E_{\beta(\ve)}[\tT_{\alpha(\ve)}] \] allowing us to handle both summands separately. 

\begin{proposition}[implying (A2)] Let $0<\alpha<\beta$ arbitrary and (by an abuse of notation) we set $\alpha(\ve)=\alpha\ve$ and $\beta(\ve)=\beta\ve$. The expected time of going from $\alpha\ve$ to $\beta\ve$ and back again without hitting $z$ is well behaved in the sense, that
	\begin{align*}
	%\lim\limits_{\ve\downarrow 0} \E_{\ve}[\widetilde{\sigma}_1] &= \lim\limits_{\ve\downarrow 0} \E_{\ve}[\sigma_1] \\
	&\lim\limits_{\ve\downarrow 0} \Bigl(\E_{\alpha\ve}[T_{\beta\ve}] + \E_{\beta\ve}[\tT_{\alpha\ve}]\Bigr) = \lim\limits_{\ve\downarrow 0} \Bigl(\E_{\alpha\ve}[T_{\beta\ve}] + \E_{\beta\ve}[T_{\alpha\ve}]\Bigr) \\
	&= \frac{2}{\sigma^2} \bigg[ \int_0^\infty \int_\alpha^\beta \exp\left( \frac{a}{\sigma^2} (1/w - 1/y) \right)  \frac{w^{b/\sigma^2}}{y^{b/\sigma^2+2}}    \,dw \,dy  \bigg] \in (0,\infty).\end{align*}
\end{proposition}

\begin{proof} 
	
We observe that for non-negative bounded and continuous functions $f$ we will have \begin{align}\label{e:exAsGreen}\E_{\alpha\ve}\biggl[\int_0^{T_{\beta\ve}} f(X_s) \, ds \biggr] = \int_0^{\beta\ve} g(\alpha\ve,y) f(y) r(y) \,dy ,\end{align} where $g$ denotes the Green kernel of $L$. In order to determine the Green kernel we calculate two solutions $u$ and $v$ of $Lw=0$:
\begin{itemize}
	\item First the constant function $u\equiv1$ is a solution and notice that the function $u$ belongs to $L^2(r(x)dx)$.
	\item $v(x):=\int_x^{\beta\ve}\frac{1}{p(w)}\,dw$ solves $Lv=0$ with the additional property that ${v(\beta\ve)=0}$. 
\end{itemize}
Therefore we conclude that the Green kernel is given by 
\begin{align} \label{e:green} 
	g(x,y)=\begin{cases}
	\frac{2}{W(v,u)}v(x)u(y)\quad &\text{if $x\geq y$}\\
	\frac{2}{W(v,u)}u(x)v(y)\quad &\text{if $x < y$},
	\end{cases}
\end{align}
where $W(v,u)=v\cdot pu'- u\cdot pv'= 1$ is the Wronskian determinant (cf. Theorem~13.21 in \cite{W03}).
%Hence, the green kernel ($\ref{e:green}$) remains in the same form with our more general $p(x)$ and $r(x)$ inserted.
Inserting $f\equiv 1$ in equation~(\ref{e:exAsGreen}) we conclude \begin{align*}
	\E_{\alpha\ve}[T_{\beta\ve}] & = 2 \ve^2 \bigg[ \int_0^\alpha \int_\alpha^{\beta}\frac{r(y\ve)}{p(w\ve)}\,dw \,dy  + \int_\alpha^{\beta} \int_y^{\beta} \frac{r(y\ve)}{p(w\ve)}\,dw  \,dy \bigg].
\end{align*}
For $\ve<\delta_0/\beta$ Lemma~\ref{l:intEst} part a) demonstrates using $ \sigma^2-M\beta\ve >\sigma^2 - M\delta_0 > \sigma^2/2 $ that  
\begin{align} \label{e:maj}
\frac{1}{y^2 \sigma^2/2 } \left( \frac{w}{y}\right)^{(a+\sigma^2)/(2\beta\sigma^2)} \exp\left( \frac{a}{\sigma^2} (1/w - 1/y) \right)  \left(\frac{w}{y}\right)^{b/\sigma^2} 
\end{align} 
is majorizing the integrand. The majorant given in \eqref{e:maj} is integrable on the domain $D= (\alpha,\beta)\times (0,\alpha) \cup \lbrace(w,y) \mid \alpha < y < \beta,\, y \le w < \beta\rbrace$. Using dominated convergence and Lemma~\ref{l:intEst} part c) we conclude 
\[
0 < \lim\limits_{\ve\downarrow 0} \E_{\alpha\ve}[T_{\beta\ve}] = \frac{2}{\sigma^2} \bigg[ \int_D \exp\left( \frac{a}{\sigma^2} (1/w - 1/y) \right)  \frac{w^{b/\sigma^2}}{y^{b/\sigma^2+2}}    \,d(w,y)  \bigg] < \infty.
\]
We now turn to $\E_{\beta\ve}[\tT_{\alpha\ve}]$. The generator of the diffusion process conditioned not to hit $z$ before hitting $\alpha\ve$ can be calculated as an $h$-transform of $L$: 
\[
L^h f = \frac{1}{2}\sigma^2(x) f'' + \left(\frac{1}{2}\left(\varepsilon b_1(x)-b_2(x)\right) + \sigma^2(x) \frac{h'(x)}{h(x)}\right)  f',
\] 
where 
\begin{equation}\label{e:harm-hitting}
h(x) := \P_{x}(T_{\alpha\ve} < T_z) = \int_x^z 1/{p}(y) \,dy \left/  \int_{\alpha\ve}^z 1/{p}(y) \,dy .\right.
\end{equation} 
$L^h$ may be rewritten in divergence form as in \eqref{e:div-form} by letting \[p^h(x) = \exp \left(\int_c^x \frac{\ve b_1(l) -b_2(l)}{\sigma^2(l)} + 2\frac{h'(l)}{h(l)} \,dl\right)\qquad \text{and} \qquad r^h(x) = \frac{p^h(x)}{ \sigma^2(x)}.
\]
Again using the corresponding Green's function we find 
\begin{align}
\E_{\beta\ve}[\tT_{\alpha\ve}] \label{e:eDownCross}&=2\ve^2\bigg[\int_\alpha^{\beta}\int_\alpha^y\frac{r^h(y\ve)}{p^h(w\ve)}\,dw \,dy + \int_{\beta}^{z/\ve}\int_\alpha^{\beta}\frac{r^h(y\ve)}{p^h(w\ve)}\,dw  \,dy  
\bigg].
\end{align}
Here we have
\[\frac{r^h(y\ve)}{p^h(w\ve)} = \frac{r(y\ve)}{p(w\ve)} \left(\frac{h(y\ve)}{h(w\ve)}\right)^2.\]

Since on the integration domains the relation $w\le y$ holds and since by \eqref{e:harm-hitting} the harmonic function $h$ is non-increasing, we get  \begin{equation}\label{e:h-quotient}
\left(\frac{h(y\ve)}{h(w\ve)}\right)^2 \leq 1.
\end{equation} 
Similar to (\ref{e:p_ezNum}) and (\ref{e:p_ezDen}) one sees that \[
\frac{h(y\ve)}{h(w\ve)} \ge \frac{h(y\ve)}{h(\alpha\ve)} = \frac{\int_{y\ve}^z 1/{p}(l) \,dl}{\int_{\alpha\ve}^z 1/{p}(l) \,dl} = 1- \frac{\int_{\alpha\ve}^{y\ve} 1/{p}(l) \,dl}{\int_{\alpha\ve}^z 1/{p}(l) \,dl} \xrightarrow[\ve\downarrow 0]{} 1. 
\] 
So applying Lemma~\ref{l:intEst} part b) and c) in order to find an integrable majorant as well as the pointwise limit we derive by Lebegue's theorem 
\[
\lim\limits_{\ve\downarrow 0} 2\ve^2\int_\alpha^{\beta}\int_\alpha^y\frac{r^h(y\ve)}{p^h(w\ve)}dw \,dy = \frac{2}{\sigma^2} \bigg[ \int_\alpha^\beta \int_\alpha^y \exp\left( \frac{a}{\sigma^2} (1/w - 1/y) \right)  \frac{w^{b/\sigma^2}}{y^{b/\sigma^2+2}}    \,dw \, dy  \bigg].
\]

We decompose the second integral in (\ref{e:eDownCross}) into two parts
\[ 
2\int_{\beta\ve}^z \int _{\alpha\ve}^{\beta\ve} \frac{r(y)}{p(w)}  \left(\frac{h(y)}{h(w)}\right)^2 \,dw\,dy = 2(I_1 + I_2), 
\] 
where 
\begin{align*}
I_1 &:= \int_{\beta\ve}^\delta \int _{\alpha\ve}^{\beta\ve}\frac{r(y)}{p(w)}    \left(\frac{h(y)}{h(w)}\right)^2 \,dw\,dy,\\
I_2 &:= \int_{\delta}^z \int _{\alpha\ve}^{\beta\ve} \frac{1}{\sigma^2(y)}\exp\left( \int_{w}^{\delta} \frac{\ve b_1(l)-b_2(l)}{\sigma^2(l)} \,dl  \right)\\
&\qquad \qquad\quad\times \exp\left( \int_{\delta}^{y} \frac{\ve b_1(l)-b_2(l)}{\sigma^2(l)} \,dl  \right)  \left(\frac{h(y)}{h(w)}\right)^2 \,dw\,dy
\end{align*} 
with $\delta := \delta_0 \wedge z/2$. We can apply the argument leading to \eqref{e:h-quotient}and Lemma~\ref{l:intEst} b) to $I_1$ since we have for $\ve < \delta/\beta \wedge \left( \sigma^4/[M(a+\sigma^2)] \right) $ in order to get the majorant 
\[ 
\frac{1}{y^2 \sigma^2/2} \cdot \frac{2y}{w} \exp\left( \frac{a}{\sigma^2} \left(\frac{1}{w} - \frac{1}{y}\right)\right) \left(\frac{w}{y}\right)^{b/\sigma^2} \left(\frac{3}{2}\right)^{b/\sigma^2+1}.
\] 
By part c) of Lemma~\ref{l:intEst} his results in 
\[
I_1\xrightarrow{\ve\downarrow 0} \int_{\beta}^\infty \int _\alpha^{\beta} \frac{1}{y^2\sigma^2 }    \exp\left( \frac{a}{\sigma^2} \left(\frac{1}{w} - \frac{1}{y}\right)\right) \left(\frac{w}{y}\right)^{b/\sigma^2}  \,dw\,dy.
\] 
The statement $I_2\to 0$ can be deduced in bounding \[ I_2 \le \int_{\delta}^z   \frac{1}{\sigma^2(y)}\exp\left( \int_{\delta}^{y} \frac{\ve b_1(l)-b_2(l)}{\sigma^2(l)} \,dl  \right) \,dy \cdot \int_{\alpha\ve}^{\beta\ve}\exp\left( \int_{w}^{\delta} \frac{\ve b_1(l)-b_2(l)}{\sigma^2(l)} \,dl  \right) \,dw   \] as product where the first factor is monotonously decreasing and bounded (e.g. set $\ve :=1$ in that expression) and the second one vanishes: \begin{align*}
&\int_{\alpha\ve}^{\beta\ve}\exp\left( \int_{w}^{\delta} \frac{\ve b_1(l)-b_2(l)}{\sigma^2(l)} \,dl  \right) \,dw \le\int_{\alpha\ve}^{\beta\ve}\exp\left( \ve \int_{\alpha\ve}^{\delta} \frac{a+Ml}{\sigma^2 l^2 - Ml^3} \,dl  \right) \,dw \\
&= (\beta-\alpha) \ve \left( \frac{\sigma^2-M\alpha\ve}{\sigma^2-M\delta} \cdot \frac{\delta}{\alpha\ve}\right)^{\ve \frac{M(a+\sigma^2)}{\sigma^4}} \exp\left( \frac{a}{\sigma^2}\left(1/\alpha - \ve/\delta \right) \right)   \to 0,
\end{align*}
where we have again used \eqref{r:taylor} in the inequality.
\end{proof}
This gives the required property of the first moment of the cycle length. We will now establish the uniform boundedness of the second moment. 
\begin{proposition}[A3] \label{p:e1A3}
	The Cycle lengths have finite second moment uniformly in $\ve>0$: 
	\[
	\limsup_{\ve\downarrow 0} \E_{\alpha\ve}[\widetilde{\sigma}_1^2] <\infty.
	\]
\end{proposition}
\begin{proof}
We show the finiteness of both $\limsup_{\ve\downarrow 0} \E_{\alpha\ve} [(T_{\beta\ve})^2]$ and $\limsup_{\ve\downarrow 0} \E_{\beta\ve}[\bigl(\tT_{\alpha\ve}\bigr)^2]$.\\

With the already calculated Green's kernel (\ref{e:green}), we use a generalized version of Kac's moments formula as stated in section 4 of \cite{LLL11} (see also \cite{FP99} for a general extensive analysis) to infer 
\begin{align*}
   \E_{\alpha\ve}[T_{\beta\ve}^2] = 2\int_{0}^{\beta\ve} g(\alpha\ve,y) \, \E_y[T_{\beta\ve}] \, r(y) \, dy.
\end{align*} 
Together with \begin{align*}
	\E_y[T_{\beta\ve}] \le 2 \ve^2 \int_{0}^{\beta} \int_{\wy}^{\beta} \frac{r(\wy\ve)}{p(\ww\ve)} \, d\ww \, d\wy
\end{align*} 
this yields 
\begin{align*}
   \E_{\alpha\ve}[T_{\beta\ve}^2] &\le 4\ve^2 \int_{0}^{\beta} \int_{y}^{\beta} \frac{r(y \ve)}{p(w\ve)} \E_{y\ve}[T_{\beta\ve}] \,dw\,dy \le 2  \biggl[ 2\ve^2 \int_{0}^{\beta} \int_{y}^{\beta} \frac{r(y \ve)}{p(w\ve)} \,dw\,dy \biggr]^2 .\\
   \xrightarrow[\ve\downarrow 0]{} & 2  \biggl[ \frac{2}{\sigma^2} \int_{0}^{\beta} \int_{y}^{\beta} \exp\left(\frac{a}{\sigma^2} \left(1/w-1/y\right)\right) \frac{w^{b/\sigma^2}}{y^{b/\sigma^2+2}} \,dw\,dy \biggr]^2,
\end{align*} 
where we used Lemma~\ref{l:intEst} c) with majorant~(\ref{e:maj}) in the last step. To see the integrability of the majorant (\ref{e:maj}) on the extended integration domain, notice that by L'H\^opital's rule 
\begin{align}\label{e:integrability}
	\lim_{y\downarrow 0} \frac{\int_y^\beta \exp\left(\frac{a}{\sigma^2} \frac{1}{w}\right) w^{\frac{b}{\sigma^2}+\frac{a+\sigma^2}{2\beta \sigma^2}} \,dw}{ \exp\left(\frac{a}{\sigma^2} \frac{1}{y}\right) y^{\frac{b}{\sigma^2}+2+\frac{a+\sigma^2}{2\beta \sigma^2}}} = \lim_{y\downarrow 0} \frac{-\exp\left(\frac{a}{\sigma^2}\frac{1}{y}\right) y^{\frac{b}{\sigma^2} + \frac{a+\sigma^2}{2\beta \sigma^2}}}{\exp\left(\frac{a}{\sigma^2}\frac{1}{y}\right) y^{\frac{b}{\sigma^2} + \frac{a+\sigma^2}{2\beta \sigma^2}}(-\frac{a}{\sigma^2} + (\frac{b}{\sigma^2}+2+\frac{a+\sigma^2}{2\beta \sigma^2})y)}
\end{align}
and therefore
\begin{displaymath}
\lim_{y\downarrow 0} \frac{\int_y^\beta \exp\left(\frac{a}{\sigma^2} \frac{1}{w}\right) w^{\frac{b}{\sigma^2}+\frac{a+\sigma^2}{2\beta \sigma^2}} \,dw}{ \exp\left(\frac{a}{\sigma^2} \frac{1}{y}\right) y^{\frac{b}{\sigma^2}+2+\frac{a+\sigma^2}{2\beta \sigma^2}}} = \frac{\sigma^2}{a}. 
\end{displaymath} 
The required integrability now follows because the integral on the left hand side of \eqref{e:integrability} is bounded near zero. 
%Analogously for the second cycle phase where we estimated the fraction of {$h$-functions} by $1$ \begin{align*}
%	\E_{\beta\ve}[\tT_{\alpha\ve}^2] \le 2  \biggl[ 2\ve^2 \int_{\alpha}^{z/\ve} \int_{\alpha}^{y} \frac{r(y \ve)}{p(r\ve)} \,dr\,dy \biggr]^2 
%\end{align*} and we split with $\delta:= \delta_0 \wedge z/2$ \begin{align*}
%	\ve^2 \int_{\alpha}^{z/\ve} \int_{\alpha}^{y} \frac{r(y \ve)}{p(r\ve)} \,dr\,dy = \ve^2 \int_{\alpha}^{\delta/\ve} \int_{\alpha}^{y} \frac{r(y \ve)}{p(r\ve)} \,dr\,dy + \int_{\delta}^{z} \int_{\alpha\ve}^{\delta} \frac{r(y )}{p(r)} \,dr\,dy + \int_{\delta}^{z} \int_{\delta}^{y} \frac{r(y )}{p(r)} \,dr\,dy.
%\end{align*}

Estimating the quotients of $h$-functions by $1$, the second moment of second cycle phase is bounded by 
\begin{align}\label{e:secondmom-down-2}
	\E_{\beta\ve}[(\tT_{\alpha\ve})^2] \le 4 \biggl[ \int_{\alpha\ve}^{\beta\ve} \int_{\alpha\ve}^{y} \frac{r(y)}{p(w)} \E_y[T_{\alpha\ve}] \,dw\,dy  + \int_{\beta\ve}^{z} \int_{\alpha\ve}^{\beta\ve} \frac{r(y)}{p(w)} \E_y[T_{\alpha\ve}] \,dw\,dy\biggr].
\end{align} 
The first integral is readily seen to be finite uniformly in $\ve>0$ by noting $\E_y[T_{\alpha\ve}] \le \E_{\beta\ve}[T_{\alpha\ve}]$. The latter implies that 
\begin{equation*}
\int_{\alpha\ve}^{\beta\ve} \int_{\alpha\ve}^{y} \frac{r(y)}{p(w)} \E_y[T_{\alpha\ve}] \,dw\,dy \leq 2 \mathbb{E}_{\beta\ve}[T_{\alpha\ve}]^2.
\end{equation*}
We have seen before that the last expectation remains bounded. To analyze the second integral on the right hand side of \eqref{e:secondmom-down-2}, we estimate on the integration domain $\alpha\ve \le w \le y \le z$: 
\begin{align}\label{e:estimatebyf}
	\frac{r(y)}{p(w)} \le \sigma^{-2}(y) \exp \left( \ve \int_{\alpha\ve}^y \frac{b_1(l)}{\sigma^2(l)} \,dl \right) \le f(y) :=\begin{cases}
	c_1 y^{-2} & \text{ for } y \le \delta:= \delta_0 \wedge z/2, \\
	c_2        & \text{ for } y \in [\delta,z],
	\end{cases}
\end{align} 
with positive constants 
\begin{align*}
	c_1 &:= \frac{2}{\sigma^2} \exp\left(\frac{3a}{\sigma^2\alpha} \right), \\
	c_2 &:= \sup_{u\in[\delta,z]} \sigma^{-2}(u)  \exp\left(\frac{3a}{\sigma^2\alpha} \right)  \exp\left(\int_\delta^z \frac{b_1(l)}{\sigma^2(l)} \, dl \right).
\end{align*}
This gives
\begin{equation*}
\begin{split}
\int_{\beta\ve}^{z} \int_{\alpha\ve}^{\beta\ve} &\frac{r(y)}{p(w)} \E_y[T_{\alpha\ve}] \,dw\,dy \leq \int_{\beta\ve}^{z} \int_{\alpha\ve}^{\beta\ve} f(y)\E_y[T_{\alpha\ve}] \,dw\,dy \\ 
&= \ve^2 \int_{\beta}^{z / \ve} \int_{\alpha}^{\beta}f(\ve y)\E_{y\ve}[T_{\alpha\ve}] \,dw\,dy = (\beta-\alpha)\ve^2 \int_{\beta}^{z/\ve} f(\ve y)\E_{y\ve}[T_{\alpha\ve}] \,dy.
\end{split}
\end{equation*}
Estimating in our formula for the expectation $\E_{y\ve}[T_{\alpha\ve}]$ the quotient $r/p$ by $f$ as above one sees that it suffices to show 
\begin{align}\label{e:last-estimate-goal}
	\limsup_{\ve\downarrow 0} \ve^4 \int_{\beta}^{z/\ve} f(y\ve)  \biggl[ \int_{\alpha}^{y} \wy \, f(\wy\ve) \,d\wy + \int_{y}^{z/\ve} y \,  f(\wy\ve) \, d\wy  \biggr] \,dy <\infty.
\end{align}
For $y \ge \delta/\ve$ we bound using \eqref{e:estimatebyf} the $y$-integrand by 
\begin{align*}
	f(y\ve)  \biggl[ \int_{\alpha}^{y} \wy \, f(\wy\ve) \,d\wy + \int_{y}^{z/\ve} y \,  f(\wy\ve) \, d\wy  \biggr] &\le c_2  \biggl[ \int_{\alpha}^{\delta/\ve} \wy \, c_1 \frac{1}{\wy^2\ve^2} \,d\wy + \int_{\delta/\ve}^{z/\ve} y \,  c_2 \, d\wy  \biggr] \\ 
	&\le  \frac{c_2c_1}{\ve^2} \ln \left(\frac{\delta}{\alpha\ve}\right) + c_2^2 z \frac{y}{\ve}.
\end{align*} 
It follows 
\begin{equation}\label{e:last-estimate-I}
\begin{split}
	&\ve^4 \int_{\delta/\ve}^{z/\ve} f(y\ve)  \biggl[ \int_{\alpha}^{y} \wy \, f(\wy\ve) \,d\wy + \int_{y}^{z/\ve} y \,  f(\wy\ve) \, d\wy  \biggr] \,dy \\
	& \le c_1 c_2  z \, \ve \ln \left(\frac{\delta}{\alpha\ve}\right) + \frac{1}{2} c_2^2  z^3\,  \ve \xrightarrow[\ve\to 0]{} 0.
\end{split}
\end{equation}
Analogously 
\begin{align*}
	\ve^4 \int_{\beta}^{\delta/\ve} f(y\ve)   \int_{\delta/\ve}^{z/\ve} y \,  f(\wy\ve) \, d\wy   \,dy \le \ve^4 \int_{\beta}^{\delta/\ve} \frac{c_1}{y^2\ve^2} c_2 y \frac{z}{\ve} \, dy = c_1c_2 z \, \ve \, \ln\left(\frac{\delta}{\beta\ve}\right) \to 0.
\end{align*} 
On the other hand we have 
\begin{align*}
	\ve^4 \int_{\beta}^{\delta/\ve} f(y\ve)  \biggl[ \int_{\alpha}^{y} \wy \, f(\wy\ve) \,d\wy + \int_{y}^{\delta/\ve} y \,  f(\wy\ve) \, d\wy  \biggr] \,dy \le c_1^2 \int_{\beta}^{\infty} \frac{1}{y^2} \biggl[ \ln\left(\frac{y}{\alpha}\right) +1  \biggr] \,dy,
\end{align*} which is a finite bound independent of $\ve$. Therefore summing the last two estimates shows that 
\begin{equation}\label{e:last-estimate-II}
\limsup_{\ve \rightarrow 0}	\ve^4 \int_{\beta}^{\delta/\ve} f(y\ve)  \biggl[ \int_{\alpha}^{y} \wy \, f(\wy\ve) \,d\wy + \int_{y}^{z/\ve} y \,  f(\wy\ve) \, d\wy  \biggr] \,dy<\infty.
\end{equation}
As \eqref{e:last-estimate-I} and \eqref{e:last-estimate-II} imply \eqref{e:last-estimate-goal} this finishes the proof.
\end{proof}

It remains to consider an arbitrary starting point $x>0$. (In other words: proving (B1) and (B2).) We first make the following preparation:\\

\begin{lemma}
	In the scaling limit $\lambda \rightarrow \infty$, $\varepsilon \rightarrow 0$ with $\lambda^2p_{\varepsilon,z}=J   \in (0,\infty)$ we have \[ p_{\ve,z} \in O(\ve^{b/\sigma^2 +1})	. \] In particular, $\lim\limits_{scaling}\frac{\ln \ve }{\lambda ^2} = 0$.
\end{lemma}
\begin{proof}
	By inequality~(\ref{e:p_ezDen}) the denominator in expression~(\ref{e:p_ez}) is bounded away from $0$. For the numerator, assuming $\ve < \delta_0/2$ we are entirely in the regime, where the approximations of coefficient functions given in Remark~\ref{r:taylor} hold. It follows \begin{align*}
		&\int_\ve^{2\ve} \exp\left(-\ve \int_{\delta_0}^y \frac{b_1(l)}{\sigma^2(l)}\, dl\right) \exp\left(\int_{\delta_0}^y \frac{b_2(l)}{\sigma^2(l)} \,dl\right)\, dy \\
		&\le \ve \int_1^{2} \left(\frac{\sigma^2-My\ve}{\sigma^2 -M\delta_0}\right)^{\ve M (a+\sigma^2)/\sigma^4} \left(\frac{\delta_0}{y\ve}\right)^{\ve M(a+\sigma^2)/\sigma^4} \\
		&\qquad\qquad \times \exp\left(\frac{a}{\sigma^2}\left(\frac{1}{y}-\frac{\ve}{\delta_0}\right)\right)  \left(\frac{y\ve}{\delta_0}\right)^{b/\sigma^2} \left(\frac{\sigma^2+M\delta_0}{\sigma^2+My\ve}\right)^{b/\sigma^2+1}    \, dy.
	\end{align*} An application of the dominated convergence theorem finishes the proof.
\end{proof}
\begin{proposition}[First part of (B1)] \label{p:B1p1Proof}
	\[	\lim_{scaling}\E_x[T_\ve \wedge T_z] = 0 \text{ for } 0 < x <z. \]
\end{proposition}
\begin{proof}
    Making use of Green's kernel $g(x,y) = \frac{2}{K\lambda^2} u(x\wedge y) v(x\vee y)$, where $K := \int_{\ve}^{z}\frac{1}{p(w)}\,dw$, $u(x) := \int_{\ve}^{x} 1/p(w)\,dw$ and $v(x):= \int_{x}^{z} 1/p(w) \,dw$ we write \begin{align*}
    	\E_x[T_\ve \wedge T_z] &= \frac{2}{K \lambda^2} \left[ v(x)\int_{\ve}^{x}\int_{\ve}^{y} \frac{r(y)}{p(w)} \,dw \, dy + u(x)\int_{x}^{z}\int_{y}^{z} \frac{r(y)}{p(w)} \,dw \, dy \right]\\
    	&\le \frac{2}{\lambda^2} \left[ \int_{\ve}^{x}\int_{\ve}^{y} \frac{r(y)}{p(w)} \,dw \, dy + \int_{x}^{z}\int_{y}^{z} \frac{r(y)}{p(w)} \,dw \, dy \right]
    \end{align*} with integrand \begin{align*}
    	\frac{r(y)}{p(w)} = \frac{1}{\sigma^2(y)} \exp\left(\ve \int_w^y \frac{b_1(l)}{\sigma^2(l)}\, dl\right) \exp\left(\int_y^w \frac{b_2(l)}{\sigma^2(l)}\,dl\right).
    \end{align*} Since on the integration domain of the second integral $w\ge y$ holds, the exponential with the $\ve$ term in it is bounded by $1$ and the integral is overall bounded. It follows that the second integral will vanish in the scaling limit. The first integral may be decomposed in \begin{align} \label{e:int0Decomp}
    	\int_{\ve}^{x}\int_{\ve}^{y} \frac{r(y)}{p(w)} \,dw \, dy = \int_{\ve}^{\delta}\int_{\ve}^{y} \frac{r(y)}{p(w)} \,dw \, dy + \int_{\delta}^{x}\int_{\ve}^{\delta} \frac{r(y)}{p(w)} \,dw \, dy + \int_{\delta}^{x}\int_{\delta}^{y} \frac{r(y)}{p(w)} \,dw \, dy
    \end{align} with $\delta:= \delta_0 \wedge x/2$. The last one is bounded, since the $\ve$-exponential is monotonically decreasing. For the second one we note \begin{align*}
    	&\int_{\delta}^{x}\int_{\ve}^{\delta} \frac{r(y)}{p(w)} \,dw \, dy \\
    	&\le \exp\left(\ve \int_{\delta}^x \frac{b_1(l)}{\sigma^2(l)}\, dl\right) \sup_{z\in[\delta,x]}  \sigma^{-2}(z)  \int_{\delta}^{x}\int_{\ve}^{\delta}  \exp\left(\ve \int_w^{\delta} \frac{b_1(l)}{\sigma^2(l)}\, dl\right)   \,dw \, dy;
    \end{align*} the remaining integrand being bounded by \[
    \left(\frac{\sigma^2-Mw}{\sigma^2-M\delta} \frac{\delta}{w}\right)^{\ve M(a+\sigma^2)/\sigma^4} \exp\left(\frac{a}{\sigma^2}\ve(w^{-1}-\delta^{-1})\right) \le \left( \frac{2\delta}{\ve}\right)^{\ve M(a+\sigma^2)/\sigma^4} \exp\left(\frac{a}{\sigma^2}\right) .
    \] For the first integral in expression~(\ref{e:int0Decomp}) we attain the estimate
    \begin{align*}
    	&\int_{\ve}^{\delta}\int_{\ve}^{y} \frac{r(y)}{p(w)} \,dw \, dy \le 2e^{a/\sigma^2} \int_{1}^{\delta/\ve}\int_{1}^{y} \frac{1}{y^2} \left(\frac{2\delta}{\ve}\right)^{\ve M(a+\sigma^2)/\sigma^4} \,dw \, dy.
    \end{align*}  For $\ve$ sufficiently small, this is \begin{align*}
    	&\le 4e^{a/\sigma^2} \int_{1}^{\delta/\ve}\int_{1}^{y} \frac{1}{y^2}  \,dw \, dy \le 4e^{a/\sigma^2} \int_{1}^{\delta/\ve} \frac{1}{y}   \, dy = 4e^{a/\sigma^2} \ln \delta + 4e^{a/\sigma^2} \ln \frac{1}{\ve}
    \end{align*} which proves the asserted limit of the product by the previously shown scaling limit $\frac{\ln \ve}{\lambda^2}\to 0$.
\end{proof}

In order to finish the proof of (B1) we first show (B2).
\begin{proposition} [B2] \label{p:B2proof}
	\[\P_x (T_\ve < T_z) \xrightarrow[\ve\to 0]{} \frac{
		\int_x^z \exp \left(-\int_y^z \frac{b_2(l)}{\sigma^2(l)} \,dl\right) \,dy}{\int_0^z \exp \left(-\int_y^z \frac{ b_2(l)}{\sigma^2(l)} \,dl\right) \,dy} \qquad  \text{ for } 0< x < z.\]
\end{proposition}
\begin{proof} We first recall that
	\begin{align*} \P_x (T_\ve < T_z) = \frac{
		\int_x^z \exp \left(\int_y^z \frac{\ve b_1(l) -b_2(l)}{\sigma^2(l)} \,dl\right) \,dy}{\int_\ve^z \exp \left(\int_y^z \frac{\ve b_1(l) -b_2(l)}{\sigma^2(l)} \,dl\right) \,dy}.
	\end{align*} 
	By an direct application of dominated convergence \begin{align}\label{e:b2Num}
	\lim_{\ve\downarrow 0} \int_x^z \exp \left(\int_y^z \frac{\ve b_1(l) -b_2(l)}{\sigma^2(l)} \,dl\right) \,dy = \int_x^z \exp \left(-\int_y^z \frac{b_2(l)}{\sigma^2(l)} \,dl\right) \,dy.
	\end{align} Since \begin{align*}
	&\int_{\ve}^{\delta_0} \exp \left(\ve \int_y^z \frac{ b_1(l) }{\sigma^2(l)} \,dl\right) \,dy \\
	&\le \exp \left(\ve \int_{\delta_0}^{z}\frac{b_1(l)}{\sigma^2(l)} \,dl \right) \int_{\ve}^{\delta_0} \left(\frac{2\delta_0}{\ve}\right)^{\ve M (a+\sigma^2)/\sigma^4} \exp(a/\sigma^2) \, dy
	\end{align*} we obtain \begin{align}\label{e:b2Den}
		\lim_{\ve \rightarrow 0} \int_\ve^z \exp \left(\int_y^z \frac{\ve b_1(l) -b_2(l)}{\sigma^2(l)} \,dl\right) \,dy = \int_0^z \exp \left(-\int_y^z \frac{ b_2(l)}{\sigma^2(l)} \,dl\right) \,dy
	\end{align} Both assertions \eqref{e:b2Num} and \eqref{e:b2Den} together imply the Proposition.
\end{proof}

We now complete the discussion of the example with
\begin{proposition}[Finishing (B1)]
	\[	\lim_{scaling}\E_x[\tT_\ve ] = 0 \text{ for } x >0. \]
\end{proposition}
\begin{proof}
%	\begin{align*}	
%	\E_{x}[\tT_{\ve}]=\frac{2}{\lambda^2}\bigg[\int_{\ve}^{x}\int_{\ve}^y\frac{r^h(y)}{p^h(r)}\,dr \, \,dy + \int_{x}^z \int_{\ve}^{x}\frac{r^h(y)}{p^h(r)}\,dr  \,dy  
%	\bigg]
%	\end{align*} with \[\frac{r^h(y)}{p^h(r)} = \frac{r(y)}{p(r)} \left(\frac{h(y)}{h(r)}\right)^2 \le \frac{r(y)}{p(r)}.\]
%	Since the first integral is just the same as in the proof before, it only remains to consider the second one. This is achieved by observing \begin{align*}
%	\int_{x}^z \int_{\ve}^{x}\frac{r(y)}{p(r)}\,dr  \,dy \le \sup_{s\in[x,z]}\sigma^{-2}(s) \int_{x}^z \int_{\ve}^{x}  \exp\left(2a_+\ve\int_\ve^z \frac{1}{\sigma^2(l)} \, dl\right)  \,dr  \,dy <\infty
%	\end{align*} with recalling $a_+ := \sup b_1$ from (E3).

    By Propositions~\ref{p:B1p1Proof} and \ref{p:B2proof}
    \begin{align*}
    	\E_x[\tT_{\ve}] = \frac{\E_x[T_\ve \ind_{\lbrace T_\ve < T_z \rbrace} ]}{\P_x(T_\ve < T_z)} \le \frac{\E_{x}[T_\ve \wedge T_z]}{\P_x(T_\ve < T_z)} \xrightarrow[\ve\downarrow 0]{} 0.
    \end{align*}
\end{proof}

\subsection{Homodyne detection of Rabi oscillation}
As is carefully described in \cite{BB17} an analysis of homodyne detection of Rabi oscillations leads to the following stochastic differential equation on the state space $\Theta=(0,2 \pi)$
\begin{equation}\label{e:RabiSDE}
d\theta_t=-\lambda^2 \sin\theta_t\bigl(1-\cos\theta_t\bigr)\,dt+\lambda\bigl(1-\cos\theta_t\bigr)dB_t
\end{equation}
Following a suggestion of \cite[sections 2.3, 6.2]{BB17} we investigate a 'linearized' version of \eqref{e:RabiSDE}, i.e. the case where in \eqref{e:mainSDE}
\[
b_1(x) = 1, \qquad  b_2(x) = b\cdot x, \qquad \sigma(x) = x^2,
\] 
with $b>0$ some positive real number. Note, in this model $\sigma^2(x)=x^4$.

\begin{remark}[Heuristics for the choice of $\alpha$ and $\beta$]
One way to guess the form of the functions $\alpha$ and $\beta$ appearing in the cycle decomposition is the following. First it is of course natural to assume that the point, where the drift changes sign does play a specific role. Therefore, let us define $\alpha(\varepsilon):=\varepsilon/b$. In order to get an idea, of how to choose $\beta(\varepsilon)$ one can e.g. first transform the stochastic differential equation using a transformation going back to at least to Feller \cite{F52}. We replace $X_t$ by $Y_t := F(X_t)$, where $F(x) := \int_{\infty}^{x} \frac{1}{\sigma(u)} \, du = -1/x$. According to It\^o's lemma the SDE then becomes 
\[
dY_t := \frac{1}{2} \left( \ve Y_t^2 + b Y_t + \frac{1}{Y_t} \right) \, dt + dB_t.
\] 
Thus we end up with diffusion process with unit diffusion coefficient. For the diffusion $X$ started from $\alpha(\varepsilon)$ to complete a cycle it has to get from $\alpha(\varepsilon)$ to  $\beta(\varepsilon)$ and back. During a downcrossing from $F(\beta(\varepsilon))$ to $F(\alpha(\varepsilon))$ one make use of the fact that the drift always points in towards  $F(\alpha(\varepsilon))$ and it turns out that the deterministic part is strong enough to get finite expectation for the part of the cycle. 
The diffusion $Y$ makes also an upcrossing from $F(\alpha(\varepsilon))$ to $F(\beta(\varepsilon))$ during a cycle of $X$. During such an upcrossing the drift in the equation of $Y$ is of order $\varepsilon$ near $\alpha(\varepsilon)$ and therefore the Brownian part has be essential to complete this part of the cycle sufficiently fast. Therefore, it seems reasonable to take $\beta(\varepsilon)=\alpha(\varepsilon)+\varepsilon^2$ as this gives
\begin{displaymath}
F(\beta(\varepsilon)-F(\alpha(\varepsilon)=\frac{b^2}{1+\varepsilon\,b}.
\end{displaymath}
The exit times of Brownian motion from bounded sets have moments of all order, thus this might be a reasonable first guess. Working with $\beta(\varepsilon)=\alpha(\varepsilon)+\varepsilon$ in contrast leads to a distance $F(\beta(\varepsilon))-F(\alpha(\varepsilon))$, which is of order $\varepsilon^{-1}$ and therefore the expected time to complete this part of the cycle can be expected to diverge with $\varepsilon \rightarrow 0$.
\end{remark}

We now show, that Theorem~\ref{t:mainThm} applies to this situation, which means that we need to prove $(A2), (A3), (B1)$ and $(B2)$ for 
\[
\alpha(\ve) := \ve/b, \qquad \beta(\ve) := \ve/b + \ve^2.
\] 
By Taylor's theorem, for $x \ge 1/b$ 
\begin{align} \label{e:taylor_le}
	\frac{1}{3x^3} - \frac{b}{2x^2} \le -\frac{b^3}{6} + \frac{b^5}{2} \left(x-\frac{1}{b}\right)^2
\end{align} 
and 
\begin{align} \label{e:taylor_ge}
	\frac{1}{3x^3} - \frac{b}{2x^2} \ge -\frac{b^3}{6} + \frac{b^5}{2} \left(x-\frac{1}{b}\right)^2 - \frac{4b^6}{3} \left(x-\frac{1}{b}\right)^3.
\end{align}
As preparation for the following proofs we start with 
\begin{lemma} \label{l:hitProb}
	For $l\ge 0$ \[
	\P_{\ve/b + l\ve^2} (T_z < T_{\ve/b}) \sim \ve^2 \exp\left(-\frac{1}{\ve^2} \frac{b^3}{6}\right)\frac{\int_{0}^{l} \exp\left(\frac{b^5}{2}x^2 \right) \,dx}{\int_{0}^{z} \exp\left(- \frac{b}{2x^2}\right) \,dx}  \qquad \text{as }\; \ve \downarrow 0 .
	\]
\end{lemma}

\begin{proof}
	\begin{align*}
	\P_{\ve/b + l\ve^2} (T_z < T_{\ve/b}) = \frac{\int_{\ve/b}^{\ve/b+l\ve^2}1/p(x)\, dx}{\int_{\ve/b}^z 1/p(x)\, dx},
	\end{align*}	 where \begin{align*}
	\frac{1}{p(x)} &= \exp\left(\ve \int_{x}^{c} \frac{1}{t^4}\, dt\right) \exp\left(b\int_{c}^x \frac{1}{t^3} \, dt\right) \\
	& = \exp\left(\frac{\ve}{3} \left(\frac{1}{x^3} - \frac{1}{c^3}\right)\right)  \exp\left(\frac{b}{2}\left(\frac{1}{c^2} - \frac{1}{x^2}\right)\right).
	\end{align*} Plugging in and reducing the fraction yields \begin{align*}
	\frac{\int_{\ve/b}^{\ve/b+l\ve^2}1/p(x)\, dx}{\int_{\ve/b}^z 1/p(x)\, dx} &= \frac{\int_{\ve/b}^{\ve/b+l\ve^2}  \exp\left(\frac{\ve}{3x^3} - \frac{b}{2x^2} \right)   \, dx}{\int_{\ve/b}^z \exp\left(\frac{\ve}{3x^3} - \frac{b}{2x^2} \right) \, dx}.
	\end{align*}
	
	For the numerator, we use the estimate~(\ref{e:taylor_le}) and obtain 
	\begin{align*} 
	&\int_{1/b}^{1/b+l\ve} \exp\left( \frac{1}{\ve^2} \left(\frac{1}{3x^3} - \frac{b}{2x^2} \right)  \right) \,dx \le \exp\left( -\frac{1}{\ve^2} \frac{b^3}{6}\right) \ve \int_{0}^{l}\exp\left(  \frac{b^5}{2} x^2   \right) \,dx.
	\end{align*} 
	Estimate~(\ref{e:taylor_ge}) gives 
	\begin{align*}
	&\int_{1/b}^{1/b+l\ve} \exp\left( \frac{1}{\ve^2} \left(\frac{1}{3x^3} - \frac{b}{2x^2} \right)  \right) \,dx \\
	&\ge \int_{1/b}^{1/b+l\ve} \exp\left( \frac{1}{\ve^2} \left(     -\frac{b^3}{6} + \frac{1}{2}b^5 \left(x-\frac{1}{b}\right)^2 - \frac{4}{3}b^6  \left(x-\frac{1}{b}\right)^3     \right)  \right) \,dx \\
	& \sim \exp\left( -\frac{1}{\ve^2}      \frac{b^3}{6} \right) \ve \int_{0}^{l} \exp  \left(    \frac{1}{2}b^5 x^2     \right)   \,dx. 
	\end{align*}
	Summarizing we arrive 
	\begin{equation}\label{e:intEpsEst}
	\int_{1/b}^{1/b+l\ve} \exp\left( \frac{1}{\ve^2} \left(\frac{1}{3x^3} - \frac{b}{2x^2} \right)  \right) \,dx \sim \exp\left( -\frac{1}{\ve^2}      \frac{b^3}{6} \right) \ve \int_{0}^{l} \exp  \left(    \frac{1}{2}b^5 x^2     \right)   \,dx. 
	\end{equation}
	For the denominator \begin{align*}
	\int_{\ve/b}^z \exp\left(\frac{\ve}{3x^3} - \frac{b}{2x^2} \right) \, dx = \int_{0}^z \ind_{\lbrace x>\ve/b \rbrace}\exp\left(\frac{\ve}{3x^3} - \frac{b}{2x^2} \right) \, dx
	\end{align*} due to $x\mapsto \ve/(3x^3) - b/(2x^2)$ being non-positive for $x \ge 2\ve/(3b)$, the integrand is bounded in-between $0$ and $1$ allowing to integrate over the limit $\ve\to 0$, which results in \[
	\lim_{\ve\to 0} \int_{0}^z \ind_{\lbrace x>\ve/b \rbrace}\exp\left(\frac{\ve}{3x^3} - \frac{b}{2x^2} \right) \, dx = \int_0^z \exp\left(-\frac{b}{2x^2}\right)\,dx.
	\] 
Thus the assertion is shown by composing these calculations.
\end{proof}

\begin{proposition}[A2] \label{p:rabiA2}
	\begin{align*}
	\lim\limits_{\ve\downarrow 0} \E_{\ve/b}[\widetilde{\sigma}_1] = \lim\limits_{\ve\downarrow 0} \E_{\ve/b}[\sigma_1] = 4b^4 \int_{0}^{\infty} \int_{0}^{1} \exp\left(\frac{b^5}{2}\left(w^2-y^2\right)\right)\,dw\,dy. \end{align*}
\end{proposition}
\begin{proof}
Using again the appropriate Green kernel we arrive at 
\begin{align}\label{e:cycLength}
\E_{\ve/b}[T_{\ve/b+\ve^2}] &= 2\ve^2 \left[ \int_{0}^{1/b} \int_{1/b}^{1/b+\ve} \frac{r(y\ve)}{p(w\ve)} \,dw \, dy + \int_{1/b}^{1/b+\ve} \int_{y}^{1/b+\ve} \frac{r(y\ve)}{p(w\ve)} \,dw \, dy \right] 
\end{align} 
where we can explicitly write
\begin{align} \label{e:intFrac}
\frac{r(y\ve)}{p(w\ve)} &= \frac{1}{(y\ve)^4} \exp\left( \frac{1}{\ve^2} \left(\left(\frac{1}{3w^3} - \frac{b}{2w^2} \right) - \left( \frac{1}{3y^3} - \frac{b}{2y^2} \right)\right) \right).
\end{align}
Observe that the right hand side of expression~\eqref{e:intFrac} factorizes in a function of $w$ and a function of $y$. 
To calculate the limit $\ve\to 0$ of the first term in \eqref{e:cycLength}, we consider the asymptotic behavior of both factors given
by the integral with respect to $y$ and $w$, respectively. We have 
\begin{align*}
&\int_{0}^{\frac{1}{b}} \frac{1}{y^4} \exp\left( \frac{1}{\ve^2} \left( - \frac{1}{3y^3} + \frac{b}{2y^2} \right) \right) \,dy 
= \exp\left(\frac{b^3}{6\ve^2}\right) \ve  \int_0^\infty (y\ve+b)^2 \exp \left(-\frac{b}{2}y^2 - \frac{\ve}{3} y^3\right) \,dy. 
\end{align*} 
Making use of (\ref{e:intEpsEst}) in order to find the asymptotic behavior of the integral with respect to $w$ and multiplying both together shows \begin{align*}
&2\ve^2 \int_{0}^{1/b} \int_{1/b}^{1/b+\ve} \frac{r(y\ve)}{p(w\ve)} \,dw \, dy  \xrightarrow[\ve\downarrow 0]{} 2b^4 \int_{0}^{\infty} \int_{0}^{1}\exp\left(  \frac{b^5}{2} \left(w^2-y^2\right)   \right) \,dw\,dy .
\end{align*}

The second summand in expression~(\ref{e:cycLength}) is analyzed with use of (\ref{e:taylor_ge}) and (\ref{e:taylor_le}). By a very similar analysis we obtain 
\begin{align}\begin{split}\label{e:expcycleI}
\E_{\ve/b}[T_{\ve/b+\ve^2}]\xrightarrow[\ve\downarrow 0]{} 
%&\xrightarrow[\ve\downarrow 0]{} \sqrt{\pi  b^3} \int_{0}^{1}\exp\left(  \frac{b^5}{2} r^2   \right) \,dr + 2b^4 \int_{0}^{1} \int_{y}^{1} \exp\left(\frac{b^5}{2} \left(r^2-y^2\right)\right) \, dr \, dy < \infty.
 &2b^4 \int_{0}^{\infty} \int_{0}^{1}\exp\left(  \frac{b^5}{2} \left(w^2-y^2\right)   \right) \,dw \,dy \\ 
 & + 2b^4 \int_{0}^{1} \int_{y}^{1} \exp\left(\frac{b^5}{2} \left(w^2-y^2\right)\right) \, dw \, dy < \infty.
\end{split}\end{align}

To infer the expected cycle length of the second phase, where the process starts from $\beta(\ve)=\ve/b +\ve^2$ and is conditioned to hit $\alpha(\ve)=\ve/b$ prior to some arbitrary level $z>\beta(\ve)$, we will again use a $h$-transform in the sense of Doob in order find the dynamics of the conditioned process.
We find  
\begin{align} \label{e:rabi1Mom}
 \E_{\frac{\ve}{b}+\ve^2}[T_{\ve/b} \mid T_{\ve/b} < T_z]  =   2\ve^2 \left[ \int_{\frac{1}{b}}^{\frac{1}{b}+\ve} \int_{\frac{1}{b}}^{y} \frac{r^h(y\ve)}{p^h(w\ve)} \,dw \, dy + \int_{\frac{1}{b}+\ve}^{\frac{z}{\ve}} \int_{\frac{1}{b}}^{\frac{1}{b}+\ve} \frac{r^h(y\ve)}{p^h(w\ve)} \,dw \, dy \right],  
 \end{align} 
where the integrand is given by 
\[
\frac{r^h(y\ve)}{p^h(w\ve)} = \left(\frac{h(y\ve)}{h(w\ve)}\right)^2 \cdot \frac{r(y\ve)}{p(w\ve)}.
\] 
We recall that the harmonic function under consideration is $h(s) := \P_s(T_{\ve/b} < T_z)$.
Let us start with the first summand. Because on the integration domain $w\le y$ holds, the estimate 
\[
\frac{r^h(y\ve)}{p^h(w\ve)} \le \frac{r(y\ve)}{p(w\ve)}
\] 
allows us use a strategy very similar to the situation of the first cycle phase. In particular, we have
\begin{align*}
&\limsup_{\ve\to 0} 2\ve^2 \int_{1/b}^{1/b+\ve} \int_{1/b}^{y} \frac{r^h(y\ve)}{p^h(w\ve)} \,dw \, dy  \le  2b^4 \int_{0}^{1} \int_{0}^{y} \exp\left(\frac{b^5}{2} \left(w^2-y^2\right)\right) \, dw \, dy.
\end{align*} 
In order to derive a matching result for the limes inferior we use our standard estimates to find 
\begin{align*}
&2\ve^2 \int_{1/b}^{1/b+\ve} \int_{1/b}^{y} \frac{r^h(y\ve)}{p^h(w\ve)} \,dw \, dy \\
&\ge 2 \int_{0}^{1}  \frac{1}{(y\ve+1/b)^4 } \exp \left( - \frac{1}{2}b^5y^2  \right) \int_{0}^{y} \left(\frac{h((y\ve+1/b)\ve)}{h((w\ve+1/b)\ve)}\right)^2 \exp \left(\frac{1}{2} b^5 w^2 - \ve \frac{4}{3}b^6y^3 \right)  \,dw \,dy.
\end{align*} 
By the bounded convergence theorem we can interchange the limit and the integrals and using Lemma~\ref{l:hitProb} to conclude 
\begin{align} \label{e:intLim}
\begin{split}
\lim\limits_{\ve\downarrow 0}  2\ve^2 \int_{1/b}^{1/b+\ve} \int_{1/b}^{y} \frac{r^h(y\ve)}{p^h(w\ve)} \,dw \, dy &= \lim\limits_{\ve\downarrow 0}  2\ve^2 \int_{1/b}^{1/b+\ve} \int_{1/b}^{y} \frac{r(y\ve)}{p(w\ve)} \,dw \, dy   \\
&=   2b^4 \int_{0}^{1} \int_{0}^{y} \exp\left(\frac{b^5}{2} \left(w^2-y^2\right)\right) \, dw \, dy. 
\end{split}
\end{align} 
Here we have used $h((l\ve+1/b)\ve)\to 1$, $l\in [0,1]$. 

We now consider the second term in equation~(\ref{e:rabi1Mom}). We rewrite this term as
\begin{equation}\label{e:last-term-1Mom}
\begin{split}
&2\ve^2 \int_{1/b+\ve}^{z/\ve} \int_{1/b}^{1/b+\ve} \frac{r^h(y\ve)}{p^h(w\ve)} \,dw \, dy = 2 \ve^4 \int_{1}^{z/\ve^2-1/(b\ve)} \int_{0}^{1} \frac{r^h((y\ve+1/b)\ve)}{p^h((w\ve+1/b)\ve)} \,dw \, dy\\
&=2 \int_{1}^{\infty} \int_{0}^{1} \ind_{\lbrace y < z/\ve^2-1/(b\ve) \rbrace} \frac{h((y\ve+1/b)\ve)^2}{h((w\ve+1/b)\ve)^2} \frac{1}{(y\ve+1/b)^4} \\
&\quad \exp\left(\frac{1}{\ve^2} \left(  \frac{1}{3(w\ve+1/b)^3} -\frac{b}{2(w\ve+1/b)^2} - \frac{1}{3(y\ve+1/b)^3} + \frac{b}{2(y\ve+1/b)^2} \right) \right)  \,dw \, dy.
\end{split}
\end{equation} 
Elementary algebra gives  
\begin{equation}\label{e:estimateexponent}
\begin{split}
&\frac{1}{\ve^2} \left(  \frac{1}{3(w\ve+1/b)^3} -\frac{b}{2(w\ve+1/b)^2} - \frac{1}{3(y\ve+1/b)^3} + \frac{b}{2(y\ve+1/b)^2} \right) \\
&= \frac{b^5(3b^2w^3\ve^2y-3b^2w\ve^2y^3+bw^3\ve + 9bw^2\ve y-9bw \ve y^2 - b\ve y^3 + 3w^2-3y^2)}{6(bw\ve+1)^3(b\ve y+1)^3}.
\end{split}
\end{equation} 
We observe that on the domain of integration in \eqref{e:last-term-1Mom} we always have $0\leq w \leq y$ and that therefore 
\begin{displaymath}
3b^2w^3\ve^2y\leq 3b^2w\ve^2y^3,\,bw^3\ve \leq b\ve y^3,\,9bw^2\ve y \leq 9bw \ve y^2.
\end{displaymath} 
Estimating the denominator $6(bw\ve+1)^3(b\ve y+1)^3$ in \eqref{e:estimateexponent} by $6$ we conclude that on the domain of integration in \eqref{e:last-term-1Mom} 
\begin{displaymath}
\exp\left(\frac{1}{\ve^2} \left(  \frac{1}{3(w\ve+1/b)^3} -\frac{b}{2(w\ve+1/b)^2} - \frac{1}{3(y\ve+1/b)^3} + \frac{b}{2(y\ve+1/b)^2} \right) \right) \leq e^{\frac{b^5}{2}(w^2-y^2)}.
\end{displaymath}
Using \eqref{e:estimateexponent} and Lebesgue's dominated convergence theorem then implies
\[
\lim\limits_{\ve \downarrow 0} 2\ve^2 \int_{1/b+\ve}^{z/\ve} \int_{1/b}^{1/b+\ve} \frac{r^h(y\ve)}{p^h(w\ve)} \,dw \, dy = 2b^4 \int_1^\infty \int_0^1 \exp\left(\frac{b^5}{2} \left(w^2-y^2\right)\right) \,dw \,dy.
\]
This gives together with \eqref{e:intLim} the required limit for the cycle phase and adding \eqref{e:expcycleI} therefore finishes the proof.
\end{proof}

\begin{proposition}[A3]
	\[\limsup_{\ve\downarrow 0} \E_{\ve/b}[\widetilde{\sigma}_1^2] <\infty.\]
\end{proposition}
\begin{proof}
	Analogously to the proof of Proposition~\ref{p:e1A3} we use Kac's moment formula and start with showing $\limsup_{\ve\downarrow 0}\E_{\ve/b}[(T_{\ve/b+\ve^2})^2] <\infty$.\\
	
	On the second double integral in \begin{align*}&\E_{\ve/b}[T_{\ve/b+\ve^2}^2] = \text{\Romannum{1}}_\ve + \text{\Romannum{2}}_\ve\\
	&=4\ve^2 \left[  \int_0^{\frac{1}{b}}\int_{\frac{1}{b}}^{\frac{1}{b}+\ve} \frac{r(y\ve)}{p(w\ve)} \E_{y\ve}[T_{\frac{\ve}{b}+\ve^2}] \,dw\,dy +  \int_{\frac{1}{b}}^{\frac{1}{b}+\ve}\int_y^{\frac{1}{b}+\ve} \frac{r(y\ve)}{p(w\ve)} \E_{y\ve}[T_{\frac{\ve}{b}+\ve^2}] \,dw\,dy  \right]\end{align*} we may estimate  $\E_{y\ve}[T_{\frac{\ve}{b}+\ve^2}] \le \E_{\frac{\ve}{b}}[T_{\frac{\ve}{b}+\ve^2}]$ and therefore finiteness follows by the convergence of the first moment shown in Proposition~\ref{p:rabiA2} as by using \eqref{e:cycLength}
	\begin{align} \label{e:greenNeglect}
	\limsup_{\ve \downarrow 0}\text{\Romannum{2}}_\ve \le \limsup_{\ve \downarrow 0} 2\left(\E_{\ve/b}[T_{\ve/b+\ve^2}]\right)^2 <\infty.
	\end{align}
	For the first integral we need to show \begin{align} \label{e:upClaim1}
	\limsup_{\ve\downarrow 0} \ve^4 \int_{0}^{\frac{1}{b}} \int_{\frac{1}{b}}^{\frac{1}{b}+\ve} \int_{0}^{y} \int_{y}^{\frac{1}{b}+\ve} \frac{r(y\ve)}{p(w\ve)} \frac{r(\widetilde{y}\ve)}{p(\widetilde{w}\ve)} \, d\widetilde{w} \, d\widetilde{y} \, dw \,dy < \infty
	\end{align} and \begin{align} \label{e:upClaim2}
	\limsup_{\ve\downarrow 0} \ve^4 \int_{0}^{\frac{1}{b}} \int_{\frac{1}{b}}^{\frac{1}{b}+\ve} \int_{y}^{\frac{1}{b}+\ve} \int_{\ty}^{\frac{1}{b}+\ve} \frac{r(y\ve)}{p(w\ve)} \frac{r(\widetilde{y}\ve)}{p(\widetilde{w}\ve)} \, d\widetilde{w} \, d\widetilde{y} \, dw \,dy < \infty.
	\end{align}
	Using (\ref{e:intFrac}) and (\ref{e:intEpsEst}) and writing $f(x) := \frac{1}{3x^3} - \frac{b}{2x^2}$ we conclude \begin{align*}&\ve^4 \int_{0}^{\frac{1}{b}} \int_{\frac{1}{b}}^{\frac{1}{b}+\ve} \int_{0}^{y} \int_{y}^{\frac{1}{b}+\ve} \frac{r(y\ve)}{p(w\ve)} \frac{r(\widetilde{y}\ve)}{p(\widetilde{y}\ve)} \, d\widetilde{w} \, d\widetilde{y} \, dw \,dy \sim \int_{0}^{1} \exp\left(\frac{b^5}{2} w^2 \right) \, dw \times \\
	&\times \frac{1}{\ve^3} \int_{0}^{\frac{1}{b}}  \int_{0}^{y} \int_{y}^{\frac{1}{b}+\ve} \frac{1}{(y\ty)^4} \exp\left(\frac{1}{\ve^2} \left(-\frac{b^3}{6} - f(y) +f(\tw) -f(\ty) \right)\right) \, d\widetilde{w} \, d\widetilde{y} \,dy.
	\end{align*} Substituting to the reciprocals, translating by $b$ and enlarging the integration domain implies \begin{align*}
	&\frac{1}{\ve^3} \int_{0}^{\frac{1}{b}}  \int_{0}^{y} \int_{y}^{\frac{1}{b}+\ve} \frac{1}{(y\ty)^4} \exp\left(\frac{1}{\ve^2} \left(-\frac{b^3}{6} - f(y) +f(\tw) -f(\ty) \right)\right) \, d\widetilde{w} \, d\widetilde{y} \,dy\\
	&\le \frac{1}{\ve^3} \int_{0}^{\infty}  \int_{y}^{\infty} \int_{-b^2\ve}^{y} \left(\frac{(y+b)(\ty+b)}{\tw+b}\right)^2 \\ &\qquad \qquad \qquad \times \exp\left(\frac{1}{\ve^2} \left(-\frac{b^3}{6} - f\left(\frac{1}{y+b}\right) +f\left(\frac{1}{\tw+b}\right) -f\left(\frac{1}{\ty+b}\right) \right)\right) \, d\widetilde{w} \, d\widetilde{y} \,dy.\end{align*}
	Since $-b^2\ve \le \tw \le y$,	\begin{align*}
	&-\frac{b^3}{6} - f\left(\frac{1}{y+b}\right) +f\left(\frac{1}{\tw+b}\right) -f\left(\frac{1}{\ty+b}\right) \le -\frac{b}{2}\ty^2 + \frac{b^5}{2}\ve^2;
	\end{align*} we continue our estimation with extending the integration domain and using Fubini's theorem to deduce with dominated convergence
	\begin{align*}
	&\exp\left(\frac{b^5}{2}\right)\int_{0}^{\infty}  \int_{y}^{\infty} \int_{-b^2}^{y} \left(\frac{(y\ve+b)(\ty\ve+b)}{\tw\ve+b}\right)^2 \exp\left(-\frac{b}{2}\ty^2\right) \, d\widetilde{w} \, d\widetilde{y} \,dy\\
	&=\exp\left(\frac{b^5}{2}\right)\int_{0}^{\infty}  \int_{0}^{\infty} \int_{-b^2}^{y} \left(\frac{(y\ve+b)(\ty\ve+y\ve+b)}{\tw\ve+b}\right)^2 \exp\left(-\frac{b}{2}\left(\ty+y\right)^2\right) \, d\widetilde{w} \, d\widetilde{y} \,dy\\
	&\le\exp\left(\frac{b^5}{2}\right)\int_{-b^2}^{\infty}  \int_{0}^{\infty} \int_{\tw}^{\infty} \left(\frac{(y\ve+b)(\ty\ve+y\ve+b)}{\tw\ve+b}\right)^2 \exp\left(-\frac{b}{2}\left(\ty+y\right)^2\right) dy \, d\ty \, d\tw\\
	&=\exp\left(\frac{b^5}{2}\right)\int_{-b^2}^{\infty}  \int_{0}^{\infty} \int_{0}^{\infty} \left(\frac{(y\ve+\tw\ve+b)(\ty\ve+y\ve+\tw\ve+b)}{\tw\ve+b}\right)^2 \exp\left(-\frac{b}{2}\left(\ty+y+\tw\right)^2\right) dy \, d\ty \, d\tw\\
	&\xrightarrow[]{\ve\to 0}\exp\left(\frac{b^5}{2}\right)b^2\int_{-b^2}^{\infty}  \int_{0}^{\infty} \int_{0}^{\infty}  \exp\left(-\frac{b}{2}\left(\ty+y+\tw\right)^2\right) dy \, d\ty \, d\tw < \infty
	\end{align*}
	which shows (\ref{e:upClaim1}). Proving (\ref{e:upClaim2}) can be performed very similar to (\ref{e:upClaim1}). Reusing the transformation $x\mapsto \frac{1}{x+b}$ yields the bound 
	\begin{align*}
	&\int_{0}^{\infty}  \int_{-b^2}^{y} \int_{-b^2}^{\ty} \left(\frac{(y\ve+b)(\ty\ve+b)}{\tw\ve+b}\right)^2 \exp\left(\frac{b}{2}\left(\tw^2-\ty^2-y^2\right)\right) \, d\widetilde{w} \, d\widetilde{y} \,dy.
	\end{align*} Noting $\tw^2-\ty^2 \le b^4$ and again using Fubini's theorem on the extended integration domain we end up with the same expression with $y$ and $\ty$ switched which is the same quantity.\\
	
	We now move on to the second cycle phase, i.e. proving \[\limsup_{\ve\downarrow 0}\E_{\ve/b+\ve^2}[T_{\ve/b}^2 \mid T_{\ve/b} < T_z] <\infty.\]
	In the spirit of \eqref{e:greenNeglect} it reduces to consider one summand and the analoga to \eqref{e:upClaim1} and \eqref{e:upClaim2}	are \begin{align*} 
	\limsup_{\ve\downarrow 0} \ve^4 \int_{\frac{1}{b}+\ve}^{\frac{z}{\ve}} \int_{\frac{1}{b}}^{\frac{1}{b}+\ve} \int_{\frac{1}{b}}^{y} \int_{\frac{1}{b}}^{\ty} \frac{r(y\ve)}{p(w\ve)} \frac{r(\widetilde{y}\ve)}{p(\widetilde{y}\ve)} \, d\widetilde{w} \, d\widetilde{y} \, dw \,dy < \infty.
	\end{align*} and \begin{align*} 
	\limsup_{\ve\downarrow 0} \ve^4 \int_{\frac{1}{b}+\ve}^{\frac{z}{\ve}} \int_{\frac{1}{b}}^{\frac{1}{b}+\ve} \int_{y}^{\frac{z}{\ve}} \int_{\frac{1}{b}}^{y} \frac{r(y\ve)}{p(w\ve)} \frac{r(\widetilde{y}\ve)}{p(\widetilde{y}\ve)} \, d\widetilde{w} \, d\widetilde{y} \, dw \,dy < \infty.
	\end{align*}
	
	Again using (\ref{e:intEpsEst}) and enlarging the domain of integration it is sufficient to show with familiar abbreviation $f(x):= \frac{1}{3x^3} - \frac{b}{2x^2}$
	%    \begin{align*}
	%    \limsup_{\ve\downarrow 0} \frac{1}{\ve^3} \int_{\frac{1}{b}}^{\infty}  \int_{\frac{1}{b}}^{\infty} \int_{\frac{1}{b}}^{\ty} \frac{1}{(y\ty)^4} \exp\left(\frac{1}{\ve^2}\left(-\frac{b^3}{6}-f(y)+f(\tr)-f(\ty)\right)\right) \, d\widetilde{r} \, d\widetilde{y}  \,dy < \infty.
	%    \end{align*}
	
	\begin{align} \label{e:downClaim1}
	\limsup_{\ve\downarrow 0} \frac{1}{\ve^3} \int_{\frac{1}{b}}^{\infty}  \int_{\frac{1}{b}}^{y} \int_{\frac{1}{b}}^{\ty} \frac{1}{(y\ty)^4} \exp\left(\frac{1}{\ve^2}\left(-\frac{b^3}{6}-f(y)+f(\tw)-f(\ty)\right)\right) \, d\widetilde{w} \, d\widetilde{y}  \,dy < \infty
	\end{align} and \begin{align} \label{e:downClaim2}
	\limsup_{\ve\downarrow 0} \frac{1}{\ve^3} \int_{\frac{1}{b}}^{\infty}  \int_{y}^\infty \int_{\frac{1}{b}}^{y} \frac{1}{(y\ty)^4} \exp\left(\frac{1}{\ve^2}\left(-\frac{b^3}{6}-f(y)+f(\tw)-f(\ty)\right)\right) \, d\widetilde{w} \, d\widetilde{y}  \,dy < \infty.
	\end{align}

	We use similar techniques resulting in
	\begin{align*}
	&\frac{1}{\ve^3} \int_{\frac{1}{b}}^{\infty}  \int_{\frac{1}{b}}^{y} \int_{\frac{1}{b}}^{\ty} \frac{1}{(y\ty)^4} \exp\left(\frac{1}{\ve^2}\left(-\frac{b^3}{6}-f(y)+f(\tw)-f(\ty)\right)\right)  \, d\widetilde{w} \, d\widetilde{y} \,dy   \\
	&= \frac{1}{\ve^3} \int_{0}^{b}  \int_{0}^{y} \int_{0}^{\ty} \left(\frac{(b-y)(b-\ty)}{b-\tw}\right)^2 \exp\left(\frac{1}{\ve^2} \left(f\left(\frac{1}{y}\right) -f\left(\frac{1}{\tw}\right) +f\left(\frac{1}{\ty}\right) \right)\right) \, d\widetilde{w} \, d\widetilde{y} \,dy.\end{align*}
	
	Due to $f$ being monotonously increasing on $[1/b,\infty)$ the difference $f(1/\ty)-f(1/\tw)\le 0$ is non-positive which combined with the fact $f(1/y)=y^3/3-by^2/2 \le -by^2/6$ provides for the estimate \begin{align*} 
	&\int_{0}^{b/\ve}  \int_{0}^{y} \int_{0}^{\ty} \left(\frac{(b-y\ve)(b-\ty\ve)}{b-\tw\ve}\right)^2 \exp\left( -\frac{b}{6}y^2 \right) \, d\widetilde{w} \, d\widetilde{y} \,dy \le \frac{b^2}{2}\int_{0}^{\infty}  y^2  \exp\left( -\frac{b}{6}y^2 \right)  \,dy 
	\end{align*}
	showing (\ref{e:downClaim1}).\\
	
	Analogously for \eqref{e:downClaim2} \begin{align*}
	&\frac{1}{\ve^3} \int_{\frac{1}{b}}^{\infty}  \int_{y}^\infty \int_{\frac{1}{b}}^{y} \frac{1}{(y\ty)^4} \exp\left(\frac{1}{\ve^2}\left(-\frac{b^3}{6}-f(y)+f(\tw)-f(\ty)\right)\right) \, d\widetilde{w} \, d\widetilde{y}  \,dy\\
	&\le \int_{0}^{b/\ve}  \int_{y}^{b/\ve} \int_{0}^{y} \left(\frac{(b-y\ve)(b-\ty\ve)}{b-\tw\ve}\right)^2 \exp\left( -\frac{b}{6}\ty^2 \right) \, d\widetilde{w} \, d\widetilde{y} \,dy\\
	&\le b^2 \int_{0}^{\infty}  \int_{y}^{\infty} y \cdot  \exp\left( -\frac{b}{6}\ty^2 \right)  \, d\widetilde{y} \,dy = \frac{3\sqrt{6 \pi b}}{4}.
	\end{align*}

\end{proof}

\begin{proposition}[First part of assertion (B1)]
	\[	\lim_{scaling}\E_x[T_{\ve/b} \wedge T_z] = 0 \text{ for } 0 < x <z. \]
\end{proposition}
\begin{proof}
	As in the first example class we are in the situation \[\E_x[T_{\ve/b} \wedge T_z] \le \frac{2}{\lambda^2} \left[ \int_{\ve/b}^{x}\int_{\ve/b}^{y} \frac{r(y)}{p(w)} \,dw \, dy + \int_{x}^{z}\int_{y}^{z} \frac{r(y)}{p(w)} \,dw \, dy \right]\] with the second integral being bounded. Using Lemma~\ref{l:hitProb} and equation~(\ref{e:intFrac}) we infer \[ \frac{2}{\lambda^2}  \int_{\ve/b}^{x}\int_{\ve/b}^{y} \frac{r(y)}{p(w)} \,dw \, dy =  2\int_{1/b}^{x/\ve}\int_{1/b}^{y} \frac{1}{y^4} \exp\left(\frac{1}{\ve^2} \left(-\frac{b^3}{6} + f(w) - f(y)\right)\right) \,dw \, dy. \] By observing \[f:[1/b,\infty)\to\R,\quad f(x):= \frac{1}{3x^3} - \frac{b}{2x^2}\] is monotonously increasing the claimed convergence is readily seen.
\end{proof}

\begin{proposition}[B2]
	\[\P_x (T_{\ve/b} < T_z) \xrightarrow[\ve\to 0]{} \frac{
		\int_x^z \exp \left(-\frac{b}{2y^2}\right) \,dy}{\int_0^z \exp \left(-\frac{b}{2y^2}\right) \,dy} \qquad  \text{ for } 0< x < z.\]
\end{proposition}
\begin{proof}
	The scale function approach leads to \[\P_x (T_{\ve/b} < T_z) = \frac{
		\int_x^z \exp \left(\frac{\ve}{3y^3}-\frac{b}{2y^2}\right) \,dy}{\int_{\ve/b}^z \exp \left(\frac{\ve}{3y^3}-\frac{b}{2y^2}\right) \,dy};\]
\end{proof} dominated convergence theorem may be applied to numerator and denominator separately finishing the proof. For the denominator observe that $\frac{\ve}{3y^3}-\frac{b}{2y^2} \le -b/(6y^2)$ holds.

\begin{proposition}[Finishing (B1)]
	\[	\lim_{scaling}\E_x[\tT_{\ve/b} ] = 0 \text{ for } x >0. \]
\end{proposition}
\begin{proof}
	As in the first example class.
\end{proof}

\section{Conclusion}
This work was mainly motivated by \cite{BB17} of M. Bauer and D. Bernard. Using a clear probabilistic heuristic we prove a version of Conjecture B under general abstract conditions and demonstrate their usability in the example sections. We believe that the approach presented above is flexible enough to cover most one-dimensional examples of interest. As already discussed in \cite{BB17} the natural question of extending the results to multi-dimensional situations remains unanswered, even though numerical simulations seem very promising in the sense that a point process could be obtained in an appropriate scaling regime. The tools and key concepts used throughout our approach appear relatively general and it would be clearly interesting to see, whether the approach of this work can be extended to higher dimensional situations. We leave this for future investigation. 

\section{Acknowledgment}
The authors thank M. Bauer, D. Bernard, A. Klump, A. Tilloy and P. Trykacz for useful discussions about the topic of this work.

%\todo*{$r$ als Integrationsvariable in $w$ umtaufen}
%\todo*{Unverwendete Nummern löschen}
%\todo*{Unverwendete Buch/Paper-Referenzen löschen (?)}
%
%\todos

\end{document}